\documentclass[%
 reprint,
 amsmath,amssymb,
 aps,
]{revtex4-2}
\usepackage[utf8]{inputenc}
\usepackage{graphicx}
\usepackage{dcolumn}
\usepackage{bbm}
\usepackage{amssymb}
\usepackage{MnSymbol}
\usepackage{amsfonts}
\usepackage{mathtools}
\usepackage{amsmath}
\usepackage{amsthm}
\usepackage{xcolor}
\usepackage{hyperref}
\usepackage{comment}
\usepackage{physics}
\usepackage{wasysym}
\usepackage{array,booktabs}
\usepackage{subcaption}

\usepackage{xcolor}

\definecolor{cA}{HTML}{D62728}
\definecolor{cB}{HTML}{1F77B4}
\definecolor{cC}{HTML}{2CA02C}
\definecolor{cD}{HTML}{FF7F0E}
\definecolor{cE}{HTML}{9467BD}
\definecolor{cF}{HTML}{8C564B}
\definecolor{cG}{HTML}{E377C2}
\definecolor{cH}{HTML}{17BECF}
\definecolor{cI}{HTML}{7F7F00}
\definecolor{cJ}{HTML}{7F7F7F}

\usepackage{tikz}
\usetikzlibrary{matrix,positioning}

\newcommand\mc[1]{\mathcal{#1}}
\newcommand{\cA}{{\mc{A}}}

\newcommand{\cC}
{{\mc{C}}}
\newcommand{\CH}{{\cC(\cH)}}

\newcommand{\CO}
{{\cC(\cO)}}

\newcommand{\cH}{{\mc{H}}}

\newcommand{\cL}{{\mc{L}}}

\newcommand{\cO}{{\mc{O}}}
\newcommand{\cP}{{\mc{P}}}
\newcommand{\cQ}{{\mc{Q}}}

\newcommand{\aP}{{\widetilde{\mc{P}}}}

\newcommand{\cS}{{\mc{S}}}

\newcommand{\ta}{{\widetilde{a}}}

\newcommand{\tx}{{\widetilde{x}}}
\newcommand{\ty}{{\widetilde{y}}}

\newcommand{\tA}{{\widetilde{A}}}

\newcommand{\tU}{{\widetilde{U}}}

\newcommand{\hf}{{\widehat{f}}}
\newcommand{\hA}{{\widehat{A}}}

\newcommand{\oL}{{\overline{L}}}

\newcommand{\ox}{{\overline{x}}}
\newcommand{\oy}{{\overline{y}}}

\newcommand{\LH}{{\mc{L}(\cH)}}
\newcommand{\LHsa}{{\mc{L}_\mathrm{sa}(\cH)}}

\newcommand{\StabS}{{\cS^\mathrm{stab}_n}}
\newcommand{\StabSS}{{\cS^\mathrm{symp}_n}}

\newcommand{\PH}{{\cP(\cH)}}
\newcommand{\PHone}{{\cP_1(\cH)}}
\newcommand{\PO}{{\cP(\cO)}}

\newcommand{\Pm}{{\cP_\mathrm{min}}}
\newcommand{\R}{{\mathbb R}}
\newcommand{\B}{{\mathbb{F}_2}}
\newcommand{\C}{{\mathbb C}}
\newcommand{\Q}{{\mathbb Q}}
\newcommand{\CP}{{\cC(\cP)}}
\newcommand{\CPm}{{\cC_\mathrm{max}(\cP)}}

\newcommand{\one}{{\mathbbm{1}}}

\newcommand{\zz}{{\mathbb{Z}}}
\newcommand{\ra}{{\rightarrow}}

\newcommand{\Sym}{{\mathrm{Sym}}}

\newcommand{\id}{{\mathrm{id}}}

\newcommand{\Hol}{{\mathrm{Hol}}}
\newcommand{\Stab}{{\mathrm{Stab}}}

\newcommand{\Lag}{{\mathrm{Lag}}}
\newcommand{\stab}{{\mathrm{stab}}}

\newcommand{\Iso}{{\mathrm{Iso}}}

\newcommand{\ConP}{{\cC(\aP_n)}}
\newcommand{\ConPP}{{\cC(\overline{\cP}_n)}}

\newcommand{\ConqP}{{\cC_q(\aP_n)}}
\newcommand{\ConS}{{\cC(\cP^\mathrm{stab}_n)}}
\newcommand{\pV}{{\cP^\mathrm{symp}_n}}
\newcommand{\ConSS}{{\cC(\pV)}}
\newcommand{\ConQ}{{\cC(\Q^{2^{n-1}})}}

\newtheorem{definition}{Definition}
\newtheorem{lemma}{Lemma}
\newtheorem{corollary}{Corollary}
\newtheorem{proposition}{Proposition}
\newtheorem{theorem}{Theorem}

\begin{document}
\preprint{APS/123-QED}

\title{Contextuality in the $n$-qubit Pauli group}
\author{Markus Frembs}
\email{markus.frembs@itp.uni-hannover.de}
\affiliation{Institut f\"ur Theoretische Physik, Leibniz Universit\"at Hannover, Appelstraße 2, 30167 Hannover, Germany}

\begin{abstract}
    The $n$-qubit Pauli group is an essential ingredient to most quantum applications, from computing and error correction to benchmarking and simulation. Despite comprising merely a discrete set of operators, it exhibits many quintessential features of quantum theory, including contextuality, which has been identified as a key resource to quantum advantage in a variety of different flavours. Here, we extend this analysis, introducing the notion of a `noncontextual property' whose nonexistence proves the Kochen-Specker theorem, similarly to and generalising common arguments based on the nonexistence of valuations. We relate this notion formally to the existence of Boolean-valued frame functions, and characterise all such frame functions in the case of the $n$-qubit Pauli group.
    
    For two qubits, we show that the Pauli group admits noncontextual properties, despite admitting no valuations. We then establish this as the only nontrivial such case with $n\geq 2$, by proving that any Boolean-valued frame function on stabiliser states is constant for more than two qubits. We also perform a similar analysis for the symplectic theory underlying the $n$-qubit Pauli group, for which nonconstant Boolean-valued frame functions exist for all $n$, yet only in restricted form. By comparison, this shows that contextuality in the n-qubit Pauli group is not only a consequence of the projective nature of the Pauli group as a representation of its underlying symplectic vector space, but of the geometry of symplectic polar spaces itself. In geometric terms, our result determines all Cameron–Liebler sets of maximal totally isotropic flats in the binary affine-symplectic space.
\end{abstract}

\maketitle

\renewcommand{\topfraction}{0.95}
\renewcommand{\textfraction}{0.05}
\renewcommand{\floatpagefraction}{0.95}
\setcounter{topnumber}{3}
\setcounter{totalnumber}{5}

\section{Introduction}

Contextuality is a key indicator for nonclassicality. The Kochen-Specker (KS) theorem \cite{Specker1960,KochenSpecker1967} asserts that (a subset of) the observables of a quantum system cannot be represented in terms of measurable functions on an underlying (classical) phase space. This foundational no-go theorem also admits a positive reading: contextuality acts as a resource in quantum computation and quantum information processing \cite{HowardEtAl2014,BudroniEtAl2022}. This motivates to analyse this resource, in particular, in settings that are relevant for applications of future quantum hardware. A natural candidate is the $n$-qubit Pauli group, which is of paramount importance to modern architectures for quantum computing, error correction, simulation etc. While contextuality of the Pauli group has been analysed before, here we substantially extend this analysis. More precisely, we define a generalised notion of noncontextual property that arises naturally from comparing the recent reformulation of the KS theorem in Ref.~\cite{Frembs2024,Frembs2025} with the common (necessary but not sufficient) criterion for noncontextuality: the existence of a valuation, that is, a map assigning operators unique spectral values in a way that respects algebraic relations between compatible, i.e. commuting ones. With respect to this generalisation, we fully characterise noncontextual properties in the $n$-qubit Pauli group (and associated stabiliser theory) as well as in its underlying symplectic theory. By comparing the two, we find that contextuality in the $n$-qubit Pauli group, in the sense of Kochen and Specker \cite{KochenSpecker1967,Frembs2024,Frembs2025}, is not solely a consequence of the group-cohomological characterisation of the Pauli group as a central extension of its underlying symplectic vector space \cite{Raussendorf2019,RaussendorfEtAl2017,OkayTyhurstRaussendorf2018,OkaySheinbaum2019,RaussendorfEtAl2023}.

\section{Noncontextual properties and frame functions}\label{sec: NC properties}

We begin by distinguishing Kochen--Specker (KS) noncontextuality from the existence of valuations. Let $\cP$ be a partial Boolean algebra (more generally, an event algebra \cite{Frembs2025}) and denote by $\CP$ its \emph{context poset}, whose objects are its finite Boolean subalgebras, ordered and identified along their overlaps.\footnote{We regard $\CP$ not merely as a poset, but as the order-preserving diagram of Boolean algebras over it; this distinction will matter in Sec.~\ref{sec: NC symplectic properties}, and is explained in more detail in App.~\ref{app: frame functions}.} We write $\CPm$ for the maximal contexts and $\Pm$ for the atoms of $\cP$. Throughout, maximal contexts are assumed to contain equally many atoms and $\cP$ to admit a dimension function $\dim:\cP\to\mathbb N$; the precise setup, and its relation with observable algebras in quantum mechanics, is reviewed in App.~\ref{app: frame functions}.

Following Refs.~\cite{Frembs2024,Frembs2025}, the obstruction to Kochen-Specker noncontextuality can be expressed in terms of identifications between different contexts.

\begin{definition}\label{def: context connection}
    Let $\cP$ be an event algebra with context poset $\CP$. A context connection $\nabla=\{\nabla_{C'C}\}_{C,C'\in\CPm}$ on $\CP$ is a family of isomorphisms, $\nabla_{C'C}:C\ra C'$, between maximal Boolean algebras $C,C'\in\CPm$ such that $\nabla_{CC'}=\nabla^{-1}_{C'C}$ and
    \begin{align*}
        \nabla_{C'C}|_{C\cap C'}=\mathrm{id}\; .
    \end{align*}
\end{definition}

Next, we evaluate context connections along context cycles, tuples of maximal contexts $\gamma=(C_0,\cdots,C_{n-1},C_n)\in(\CPm)^{n+1}$, where cyclicity means $C_n=C_0$.

\begin{definition}\label{def: contextual holonomy}
    Let $\cP$ be an event algebra, $\CP$ its context poset, and let $\nabla$ be a context connection on $\CP$. Then its \emph{holonomy group at $C_0\in\CPm$} is defined as,
    \begin{align}
        \Hol_{C_0}(\nabla)
        :=\langle P_\gamma(\nabla)\mid\gamma\ \mathrm{context\ cycle}\rangle
        \subset\Sym(C_0)\; ,
    \end{align}
    where $P_\gamma(\nabla)=\circ_{i=0}^{n-1}\nabla_{C_{i+1}C_i}$ for $\gamma=(C_0,\cdots,C_{n-1},C_n=C_0)$ and $\Sym(C_0):=\Sym(\Pm(C_0))\cong S_{|\Pm(C_0)|}$ with $S_n$ the permutation group on $n$ elements.
\end{definition}

As shown in Refs.~\cite{Frembs2024,Frembs2025}, KS noncontextuality is equivalent to the existence of a context connection satisfying
\begin{align}\label{eq: flatness}
    \Hol_{C_0}(\nabla)=\{\id\}\; .
\end{align}
By contrast, a valuation on $\cP$ requires only that one atom $p_0\in\Pm$ can be identified consistently, equivalently that the holonomy fixes this atom (see App.~\ref{app: frame functions}), that is,
\begin{align*}
    \Hol_{C_0}(\nabla)\subseteq\Stab(p_0)\; .
\end{align*}
This observation suggests relaxing the latter condition from atoms to arbitrary events.

\begin{definition}\label{def: noncontextual property}
    Let $\cP$ be an event algebra with context poset $\CP$. Then $(0,\one\neq) p_0\in C_0\in\CPm$ is called a \emph{(nontrivial) noncontextual property of $\CP$} if there exists a context connection $\nabla$ on $\CP$ such that
    \begin{align}\label{eq: generalised stabiliser condition}
        \Hol_{C_0}(\nabla)
        \leq\Stab(p_0)\; .
    \end{align}
\end{definition}
Valuations are the special case where $p_0$ is an atom.

We define a Boolean-valued frame function of weight $k$ as a finitely additive map $f:\cP\to\mathbb N_0$ such that
\begin{align*}
    f(\Pm)\subseteq\{0,1\}\; ,\qquad
    f(\one)=k\; ,
\end{align*}
where $\one$ denotes the unit in $\cP$. Equivalently, we may regard it as a map $f:\Pm\ra\{0,1\}$ selecting exactly $k$ atoms in every maximal context. Noncontextual properties and Boolean-valued frame functions are equivalent.

\begin{theorem}\label{thm: NP=2FP}
    Let $\cP$ be an event algebra and $\CP$ its context poset. Then $\CP$ admits a nontrivial noncontextual property $p\in\cP$ if and only if there exists a nonconstant Boolean-valued frame function of weight $\dim(p)$ on $\cP$.
\end{theorem}

\begin{proof}
    We provide the details in App.~\ref{app: frame functions}.
\end{proof}

Characterising noncontextual properties is thus equivalent to classifying Boolean-valued frame functions. The generalisation beyond valuations is nontrivial in both directions: the two-qubit Mermin--Peres square \cite{Mermin1990,Peres1991,Mermin1993} admits no valuation, yet it admits the rank-$2$ noncontextual properties illustrated in Fig.~\ref{fig: the noncontextual MP-square}; conversely, Thm.~\ref{thm: symplectic Boolean frame functions} below shows that the existence of valuations does not imply the existence of arbitrary noncontextual properties.

\begin{figure}[ht]
    \centering
    
    \begin{minipage}[t]{0.45\textwidth}
    \centering
    \scalebox{.8}{
    \begin{tikzpicture}[
        mpcell/.style={
            draw,
            minimum width=1.9cm,
            minimum height=1.15cm,
            align=center,
            font=\large
        },
        rowlabel/.style={font=\large\bfseries},
        collabel/.style={font=\large\bfseries},
        bluecell/.style={mpcell, fill=blue!12, draw=blue!65!black, line width=0.9pt}
    ]
    
    \node[collabel] at (0,1) {$C_1$};
    \node[collabel] at (1.9,1) {$C_2$};
    \node[collabel] at (3.8,1) {$C_3$};
    
    \node[rowlabel] at (-1.4,0) {$R_1$};
    \node[rowlabel] at (-1.4,-1.15) {$R_2$};
    \node[rowlabel] at (-1.4,-2.30) {$R_3$};
    
    \node[mpcell]    at (0,0)       {$X_1$};
    \node[bluecell]  at (1.9,0)     {$X_2$};
    \node[mpcell]    at (3.8,0)     {$X_1X_2$};
    
    \node[mpcell]    at (0,-1.15)   {$Y_2$};
    \node[mpcell]    at (1.9,-1.15) {$Y_1$};
    \node[bluecell] at (3.8,-1.15) {$Y_1Y_2$};
    
    \node[bluecell] at (0,-2.30)   {$X_1Y_2$};
    \node[mpcell]     at (1.9,-2.30) {$Y_1X_2$};
    \node[mpcell]     at (3.8,-2.30) {$Z_1Z_2$};
    
    \end{tikzpicture}
    }
    
    \vspace{0.6em}
    
    \small
    
    \end{minipage}
    \hfill
    \caption{The highlighted mutually anti-commuting Pauli operators determine the rank-$2$ projectors $\Pi_P=\frac{1}{2}(\one+P)$ onto their $+1$-eigenspaces.}
    \label{fig: the noncontextual MP-square}
\end{figure}

Motivated by this distinction, we turn to analyse Def.~\ref{def: noncontextual property} using the $n$-qubit Pauli group as a case study. More precisely, we classify all Boolean-valued frame functions for the $n$-qubit Pauli group and its stabiliser theory, and compare the result with the underlying symplectic theory. The results are summarised in Tab.~\ref{tab: summary}.

\begin{table}[t!]
\caption{Nonconstant Boolean-valued frame functions $f:\cP_{\min}\to\{0,1\}$, equivalently nontrivial noncontextual properties (Thm.~\ref{thm: NP=2FP}), for the $n$-qubit Pauli group and its underlying symplectic theory (Thm.~\ref{thm: stabiliser frame functions} and Thm.~\ref{thm: symplectic Boolean frame functions}).}
\label{tab: summary}
\footnotesize
\setlength{\tabcolsep}{3pt}

\begin{subtable}{\columnwidth}
\centering
\begin{tabular}{c|ccc}
    \toprule
    & $n=1$ & $n=2$ & $n\geq3$ \\
    \midrule
    weights & $\{1\}$ & $\{2\}$ & none \\
    $f(\neq 0,1)$ & valuations & $f_{q,s}$ & none \\
    &  & (Prop.~\ref{prop: noncontextual MP-square}) & (Thm.~\ref{thm: stabiliser frame functions}) \\
    \bottomrule
\end{tabular}
\caption{stabiliser theory, $\cP=\cP^\mathrm{stab}_n$} 
\end{subtable}\par\medskip\begin{subtable}{\columnwidth}
\centering
\begin{tabular}{c|ccc}
    \toprule
    & $n=1$ & $n=2$ & $n\geq3$ \\
    \midrule
    weights & $\{1\}$ & $\{1,2,3\}$ & $\{1,2,2^n-2,2^n-1\}$ \\
    $f(\neq 0,1)$ & valuations & $f_{\alpha,c}$, $f_{q,s}$ & $f_{\alpha,c}$, $f_{q,\alpha}$ \\
    & & (Lm.~\ref{lm: compatible characters}, Lm.~\ref{lm: noncontextual property in symplectic stabiliser poset}) & (Thm.~\ref{thm: symplectic Boolean frame functions}) \\
    \bottomrule
\end{tabular}
\caption{symplectic theory, $\cP=\cP^\mathrm{symp}_n$}
\end{subtable}
\end{table}

\section{Noncontextual properties in the $n$-qubit Pauli group}\label{sec: NC properties Pauli group}

We first analyse Def.~\ref{def: noncontextual property} for Hermitian $n$-qubit Pauli operators $\cO=\pm\aP_n$, where $\aP_n=\{\one,X,Y,Z\}^{\otimes n}$. Explicitly, we will use the Weyl representation (see e.g. Ref.~\cite{deBeaudrap2013}),
\begin{align}\label{eq: Weyl operators}
    W_v
    :=i^{a\cdot b}\otimes_{i=1}^nZ^{a_i}_iX^{b_i}_i\; ,
\end{align}
for all $v=(a,b)\in V=\zz^{2n}_2$ and $q_W(v):=a\cdot b:=\sum_{i=1}^na_ib_i$. It follows that $\omega(v,v'):=a\cdot b'-a'\cdot b$ defines a symplectic (that is, nondegenerate, alternating, bilinear) form $\omega:V\times V\ra\zz_2$ on $V$ such that $W_vW_w=(-1)^{\omega(v,w)}W_wW_v$ for all $v,w\in V$, and $W_vW_w=(-1)^{\beta_W(v,w)}W_{v+w}$ whenever $\omega(v,w)=0$, in which case $\beta_W=\frac{\gamma_W}{2}$ for $\gamma_W:V\times V\ra\zz_4$ the $2$-cocycle of the Weyl section $v\ra W_v$ in Eq.~(\ref{eq: Weyl operators}) (of the central extension $1\ra\zz_4\ra\cP_n\ra V\ra 1$).\footnote{The Hermitian Weyl section determines the $\zz_4$-valued cocycle $\gamma_W$ of the central extension. On compatible pairs, it is even and determines the $\zz_2$-valued multiplication phase $\beta_W$, which is itself a $2$-cocycle of the cochain complex of the commutativity structure of the full Pauli group, as used in contextuality proofs 
\cite{Raussendorf2019,RaussendorfEtAl2017,OkayTyhurstRaussendorf2018,OkaySheinbaum2019,RaussendorfEtAl2023}; the two cohomology obstructions are distinct but closely related, in particular, for $n>1$ both $[\gamma_W]\neq0\neq[\beta_W]$; moreover, $\beta_W$ determines $\gamma_W$ up to complex conjugation, upon which $\beta_W\mapsto\beta_W$ yet $\gamma_W\mapsto\gamma_W+2\omega$. Finally, note that the choice of cocycle $\beta_W/\gamma_W$ is without loss of generality: re-phasing $W_v\mapsto (-1)^{s(v)}W_v$ by $s:V\ra\zz_2$, together with a simultaneous change $\lambda_U\mapsto\lambda'_U=\lambda_U+s|_U$, leaves the projectors in Eq.~(\ref{eq: symplectic codespace projectors}) invariant.}

We denote by $\cA=\cA(\aP_n):=\{A<\cP_n\mid-\one\notin A\subset\pm\aP_n, A\ \mathrm{Abelian}\}$ the set of stabiliser subgroups in $\cP_n$, where $\cP_n=\langle\aP_n\rangle=Z(\cP_n)\cdot\aP_n$ and $Z(\cP_n)=\{\pm i,\pm 1\}$.\\

\textbf{Contexts from stabiliser subgroups.} In order to evaluate Eq.~(\ref{eq: generalised stabiliser condition}) on the poset of stabiliser subgroups, we need to change perspective from generating Pauli operators to the projections spanning a context. To do so, for every Abelian stabiliser subgroup $A\in\cA$ and linear character\footnote{Let $A^*:=\{\chi:A\ra\zz_2\mid\chi\mathrm{\ additive}\}$, and note that $A^*\cong\hA:=\{\chi:A\ra U(1)\mid\chi\mathrm{\ multiplicative}\}$ for $\hA$ the space of characters.} $\chi:A\ra\zz_2$, define the (codespace) projector
\begin{align}\label{eq: codespace projectors}
    \Pi_{A,\chi}
    :=\frac{1}{|A|}\sum_{P\in A}(-1)^{\chi(P)}P\; .
\end{align}
$\Pi_{A,\chi}$ is rank-$1$ and corresponds to a stabiliser state if and only if $A$ is maximal Abelian. More generally, if $\tA\in\cA$ has dimension $k$, then $\Pi_{\tA,\chi}$ has rank $2^{n-k}$. It is useful to translate Eq.~(\ref{eq: codespace projectors}) into the symplectic setting of $(V=\zz^{2n}_2,\omega)$. Recall that a subspace $W\subset V$ is called \emph{isotropic}, if the restriction of the symplectic form vanishes on $W$, that is, $\omega(w,w')=0$ for all $w,w'\in W$. We denote the space of isotropic subspaces of $V$ by $\Iso(V)$. Maximal isotropic subspaces are called \emph{Lagrangian}, have dimension $n$ and are denoted by $\Lag(V)$. With respect to the Weyl representation in Eq.~(\ref{eq: Weyl operators}), Eq.~(\ref{eq: codespace projectors}) then becomes
\begin{align}\label{eq: symplectic codespace projectors}
    \Pi_{U,\chi}
    :=\frac{1}{|U|}\sum_{v\in U}(-1)^{\chi(v)+\lambda_U(v)}W_v\; ,
\end{align}
where $\lambda_U$ is any solution to the equation $\lambda_U(v+w)=\lambda_U(v)+\lambda_U(w)+\beta_W(v,w)$ for $\beta_W$ the $2$-cocycle of the Weyl representation in Eq.~(\ref{eq: Weyl operators}).\footnote{In group-cohomological terms, $\lambda$ has coboundary the $2$-cocycle of the Weyl section, that is, $\delta\lambda_U=\beta_W|_{U\times U}$. We will mostly suppress the label $\lambda$ in our notation.} In particular, $\{\chi+\lambda_U\}_{\chi\in U^*}$ is an affine torsor over $U^*$. In other words, $\Pi_{U,\chi}$ is the projector onto the joint $\chi(v)+\lambda_U(v)$-eigenspaces of Paulis $W_v$ with $v\in U$. With respect to the Weyl representation, we may thus identify stabiliser subgroups $A^{\lambda_L}=\{(-1)^{\lambda_L(v)}W_v\mid v\in L\}\cong L$ with their uniquely associated Lagrangian subspaces $L$.

Now, consider the set of contexts of the form
\begin{align}\label{eq: Pauli context}
    C_U
    :=\mathrm{span}_\B\big\{\{\Pi_{U,\chi}\}_{\chi\in U^*}\big\}
\end{align}
for every isotropic subspace $U\in\Iso(V)$. Under inclusion, these define a partial order,
\begin{align}\label{eq: Pauli context poset}
    \ConP(=\cC(\cP(\aP_n)))
    :=\Big(\{C_U\}_{U\in\Iso(V)},\subseteq\Big)\; .
\end{align}
By construction, contexts in $\ConP$ correspond bijectively with isotropic subspaces, and $\tU\subset U$ if and only if $C_\tU\subset C_U$ for $\tU,U\in\Iso(V)$, that is, $\ConP$ (is the order-preserving diagram that) assigns a Boolean algebra $C_U$ to every element of its underlying poset $(\Iso(V),\subseteq)$.

Note that nontrivial minimal contexts of $\ConP$ are of the form $C_v=\mathrm{span}_\B(\Pi^\pm_v,\one)$ (with $\mathrm{span}_\C(C_v)=\{W_v,\one\}''$), for $\Pi^\pm_v$ the rank-$2^{n-1}$ projector onto the $\pm 1$-eigenspace of $W_v$. Notably, $\ConP$ contains far fewer contexts than those generated by all stabiliser states. For instance, $C=\mathrm{span}_\B(\dyad{00},\dyad{01},\dyad{1+},\dyad{1-})$, is generated by stabiliser states, but not contained in $\ConP$, as the first two projections have stabiliser $\langle Z_1,Z_2\rangle$, while the latter have stabiliser group $\langle Z_1,X_2\rangle$.\\

\textbf{Stabiliser frame functions.} We also consider the context poset generated by all stabiliser states $\StabS\subset\cP_1(\C^{2^n})\cong\mathbb{CP}^{2^n-1}$. Since every stabiliser state is uniquely defined by its stabiliser group $A\in\cA$, equivalently its associated Lagrangian subspace $L\in\Lag(V)$ and character $\chi\in L^*$, we have (with Eq.~(\ref{eq: symplectic codespace projectors})),
\begin{align}\label{eq: stabiliser projectors}
    \StabS
    =\{\Pi_{L,\chi}\mid L\in\Lag(V),\chi\in L^*\}\; .
\end{align}
More generally, stabiliser states generate an event algebra $\cP^\mathrm{stab}_n$, whose maximal Boolean algebras are of the form
\begin{align}\label{eq: stab contexts}
    \mathrm{span}_\B\big\{\{\Pi_{\psi_i}\}_{i=1}^{2^n}\mid\psi_i\in\StabS,\ \psi_i\perp\psi_j\ \mathrm{if}\ i\neq j\big\}\; .
\end{align}
As we show in Ref.~\cite{Frembs2026_unextendibility}, $\cP^\mathrm{stab}_n$ is not a partial Boolean algebra. In particular, it differs from the partial Boolean algebra generated by all stabiliser projectors for $n\geq 4$, since there exist sets of orthogonal stabiliser states that cannot be extended to a stabiliser basis.

Clearly, we also have $\ConP\subset\ConS$; in turn, frame functions are more constrained if finite additivity is imposed with respect to $\ConS$ as opposed to $\ConP$ only. The various diagrams are summarised in Tab.~\ref{tab: posets}.

\begin{table}[t]
    \centering
    \renewcommand{\arraystretch}{1.3}
    \begin{tabular}{l|c|c}
      & \begin{tabular}{@{}c@{}}contexts from\\ isotropic subspaces\end{tabular}
      & \begin{tabular}{@{}c@{}}maximal contexts \\ from orthogonal \\ stabiliser states\end{tabular} \\[2pt]
    \hline
    Pauli group
      & $\ConP$ & $\subset\ConS$ \\
    symplectic theory
      & $\ConPP$   & $\subset\ConSS$ \\
    \hline
    underlying poset & $\cong(\Iso(V),\subseteq)$ &
\end{tabular}
\caption{The four context posets we consider. Horizontally, the posets on the right contain strictly more contexts (for $n\geq 2$), hence impose strictly more constraints on frame functions. Vertically, the underlying posets coincide on the left (with $(\Iso(V),\subseteq$)), yet differ as diagrams for $n>1$.}
\label{tab: posets}
\end{table}

Our first main result rules out general noncontextual properties in the $n$-qubit Pauli group.

\begin{theorem}\label{thm: stabiliser frame functions}
    For $n\geq 3$, $\ConP$ and thus $\ConS$ admit no nontrivial noncontextual properties. In particular, every frame function $f:\StabS\ra\{0,1\}$ for $n\geq 3$ is constant.
\end{theorem}

\begin{proof}
    We provide the proof in App.~\ref{app: classification}.
\end{proof}

\textbf{Quaternionic valuations.} As an application of Thm.~\ref{thm: stabiliser frame functions}, we prove that no valuations exist in quaternionic quantum mechanics of dimension $2^{n-1}$.

\begin{corollary}\label{cor: no quaternionic valuations}
    For $n\geq 3$, $\cP(\Q^{2^{n-1}})$ admits no valuations.
\end{corollary}

\begin{proof}
    We provide the proof in App.~\ref{app: quaternionic valuations}.
\end{proof}

Cor.~\ref{cor: no quaternionic valuations} is of course not new, and known to hold for arbitrary dimensions \cite{Varadarajan1985,MorettiOppio2018}. However, unlike the general result which uses Gleason's theorem, in the special case of $n$-qubit systems our result holds already under the restriction to the discrete context stabiliser poset.

\section{Noncontextual properties in symplectic qubit theory}\label{sec: NC symplectic properties}

Recall that the space of isotropic subspaces defines a partial order $(\Iso(V),\leq)$, independent of any phase convention involved in picking Hermitian Pauli representatives, while $\ConP$ assigns Boolean algebras to it using the phase convention in the Weyl section. In this section, we study contextuality in the projective, purely symplectic setting, by trivialising the cocycle $\beta_W$.

Explicitly, instead of the Weyl operators in Eq.~(\ref{eq: Weyl operators}), let $\{e_v\}_{v\in V}$ denote an orthogonal basis in the (Bloch) space $\R[V]$, and for $\chi\in U^*$ and $U\in\Iso(V)$ define the vectors
\begin{align}\label{eq: symplectic projectors}
    \pi_{U,\chi}
    :=\frac{1}{|U|}\sum_{v\in U}(-1)^{\chi(v)}e_v\; .
\end{align}
Let $\StabSS:=\{\pi_{L,\chi}\mid L\in\Lag(V),\chi\in L^*\}$, define contexts $c_U:=\mathrm{span}_\B\{\{\pi_{U,\chi}\}_{\chi\in U^*}\}$ for $U\in\Iso(V)$ and denote by $\ConPP$ the (order-preserving diagram of) all such Boolean algebras over $(\Iso(V),\leq)$. Here, $\overline{\cP_n}=\cP_n/Z(\cP_n)$ and $\ConPP\cong(\cA(\cP_n/Z(\cP_n),\leq))$, where $\cA(\cP_n/Z(\cP_n)):=\{A/Z(\cP_n)\mid Z(\cP_n)<A<\cP_n,A\ \mathrm{Abelian}\}$ denotes Abelian subgroups of $\cP_n$ up to phases. However, $\ConP\ncong\ConPP$ for $n\geq 2$ as diagrams (see Lm.~\ref{lm: different posets} in App.~\ref{app: symplectic contextuality}).

Moreover, orthogonal vectors with respect to
\begin{align}\label{eq: stabiliser orthogonality}
    \langle\pi_{L,\chi},\pi_{M,\chi'}\rangle=0
    \quad:\Longleftrightarrow\quad
    \chi_{L\cap M}\neq\chi'_{L\cap M}\; ,
\end{align}
define an abstract stabiliser poset, $\ConSS$, whose maximal contexts consist of $2^n$ orthogonal vectors, analogously to $\ConS$. Clearly, $\ConPP\subset\ConSS$.

We analyse these context posets with respect to Def.~\ref{def: noncontextual property}. First, $\ConPP$ clearly admits valuations in the form of linear functionals $\alpha:V\ra\zz_2$, as $\{\alpha|_L\}_{L\in\Lag(V)}$ defines a collection of characters that is compatible under coarse-graining (see Lm.~\ref{lm: character restriction}). Yet, there are also valuations not of this form. To see this, recall that a \emph{quadratic refinement $q:V\ra\zz_2$ of $\omega$} is defined by the relation
\begin{align}\label{eq: quadratic refinement}
    q(v)+q(w)+q(v+w)=\omega(v,w)\mod 2\; .
\end{align}
Then, since $\omega|_L=0$, $q|_L\in L^*$ also defines a family of compatible characters in every Lagrangian subspace. This holds in particular for the quadratic refinement $q_W$ in Eq.~(\ref{eq: Weyl operators}), and thus also for all combinations $\alpha+q_W$.

Quadratic refinements of $\omega$ are characterised by their \emph{Witt index},\footnote{Equivalently, they are defined by their \emph{Arf invariant} \cite{Arf1941}.} defined as the largest dimension of an isotropic subspace in $q^{-1}(0)$, which is either $n$ or $n-1$ \cite{Taylor1992}. In addition to valuations, $\ConPP$ admits generalised noncontextual properties (see Lm.~\ref{lm: noncontextual symplectic properties} and Lm.~\ref{lm: noncontextual property in symplectic stabiliser poset} in App.~\ref{app: symplectic contextuality}), corresponding to quadratic refinements of Witt index $n-1$. Moreover, these noncontextual properties extend to the abstract stabiliser poset. Yet, surprisingly, besides these two cases no other noncontextual properties exist. Consequently, not only $\cP_n$ but already its underlying symplectic theory is contextual.

\begin{theorem}\label{thm: symplectic Boolean frame functions}
    For all $n\geq 3$, nonconstant Boolean-valued frame functions over $\ConPP$ and over $\ConSS$ coincide, and every such frame function $f:\cS^\mathrm{symp}_n\ra\{0,1\}$ with $V=\zz^{2n}_2$ is either of the form,
    \begin{align*}
        f_{\alpha,c}(\pi_{L,\chi})
        &=\begin{cases}
            1&\mathrm{if}\ \chi=(\alpha+cq_W)|_L\\
            0&\mathrm{if}\ \chi\neq(\alpha+cq_W)|_L
        \end{cases} \\
        f_{q,\alpha}(\pi_{L,\chi})
        &=\begin{cases}
            1&\mathrm{if}\ \chi|_{L_q}=\alpha|_{L_q}\\
            0&\mathrm{if}\ \chi|_{L_q}\neq\alpha|_{L_q}
        \end{cases}\; ,
    \end{align*}
    or their complements $\overline{f}=1-f$, where $\alpha:V\ra\zz_2$ is a linear functional, $c\in\zz_2$, $q_W(v)=\sum_{i=1}^na_ib_i$ is the quadratic refinement in Eq.~(\ref{eq: Weyl operators}) and $q:V\ra\zz_2$ a quadratic refinement of Witt index $n-1$ with $L_q=\ker(q|_L)=q^{-1}(0)\cap L$.
\end{theorem}

\begin{proof}
    We provide the proof in App.~\ref{app: symplectic contextuality}.
\end{proof}

From a purely mathematical perspective, Thm.~\ref{thm: symplectic Boolean frame functions} fully characterises affine-symplectic Cameron-Liebler sets in the binary case \cite{GuoWan2023}. We discuss this and comment on the qudit generalisation of our result in App.~\ref{app: CL correspondence}.

\section{Conclusion}\label{sec: conclusion}

It is well-known that the $n$-qubit Pauli group is contextual - in the sense that there exist no valuations on (Hermitian) Pauli operators for $n\geq 2$ \cite{Mermin1990,Peres1991,Mermin1993,Arkhipov2012,TrandafirLisonekCabello2022,MullerGiorgetti2025}. Here, we generalised this assertion by ruling out noncontextual properties for $n\geq 3$; leaving the single qubit case, as well as two-qubit case in Fig.~\ref{fig: the noncontextual MP-square} as the only exception.

By contrast, the underlying symplectic theory clearly does admit valuations in the form of linear functionals and quadratic refinements. Moreover, it also admits generalised noncontextual properties, yet these remain strongly restricted; in particular, the symplectic theory is also contextual. The comparison between the stabiliser and symplectic theories therefore illustrates that contextuality in the $n$-qubit Pauli group is \emph{not} solely a consequence of it being a projective representation of its underlying symplectic vector space (equivalently, the group-cohomological characterisation in Ref.~\cite{Raussendorf2019,RaussendorfEtAl2017,OkayTyhurstRaussendorf2018,OkaySheinbaum2019,RaussendorfEtAl2023}), but of the geometry of symplectic polar spaces more generally.

Our results thus naturally fit within and extend the broader body of work identifying the $n$-qubit Pauli group with symplectic polar spaces \cite{SanigaPlanat2007,PlanatSaniga2008,HavlicekOdehnalSaniga2009,Thas2009,LevaySanigaVrana2008,HolweckSanigaLevay2014,Frembs2026_unextendibility,FrembsHoehnNataleWever2026a,FrembsHoehnNataleWever2026b}. In fact, Thm.~\ref{thm: symplectic Boolean frame functions} characterises all Cameron–Liebler sets of maximal totally isotropic flats in the binary affine-symplectic space \cite{GuoWan2023}.

\section*{Acknowledgments.} I thank Lukas Hantzko, Selman Ipek, Robert Raussendorf and Atak Talay Yucel for discussions.

OpenAI's GPT-5.6 Sol model found the Lagrangian certificates in the proof of Lm.~\ref{lm: k=4, symplectic} and pointed out the close relation with the work of Ref.~\cite{GuoWan2023}.

\bibliography{bibliography}

\appendix
\onecolumngrid

\appendix

\section{Context posets, Kochen-Specker contextuality and proof of Thm.~\ref{thm: NP=2FP}}\label{app: frame functions}

We review the formal setting underlying Sec.~\ref{sec: NC properties}, Kochen-Specker noncontextuality and the notion of noncontextual properties (see Def.~\ref{def: noncontextual property}). We first recall how observable algebras give rise to event algebras and context posets, then relate classical embeddings and valuations to the holonomy group (see Def.~\ref{def: contextual holonomy}), and finally prove Thm.~\ref{thm: NP=2FP}.\\

\textbf{Event algebras, context connections and contextual holonomy.} Let $\cO\subset\LHsa$ with $\one\in\cO$ be a set of Hermitian operators on a (finite-dimensional) Hilbert space $\cH$. We collect the observables in $\cO$ into commuting ones, and without loss of generality consider their algebraic closure: for any subset $\one\in S\subset\cO$ of mutually commuting observables, define the \emph{context generated by $S$} as the commutative algebra,
\begin{align}\label{eq: context}
    C(S)
    :=\{S\}''\; ,
\end{align}
where $\{S\}':=\{a\in\LH\mid[a,S]=0\}$ is the commutant of $S\subset\LH$. If $S$ is finite (as in the cases below), then
\begin{align*}
    C(S)
    =\mathrm{span}_\C\Big\{\prod_{O\in S}O^{m_O}\mid m_O\in\mathbb{N}_0\Big\}
    =\mathrm{span}_\C\{\cP(S)\}\; ,
\end{align*}
where $\cP(S)$ denotes the projectors in the spectral decompositions of the operators in $S$.\footnote{For a thorough comparison of these notions, see Ref.~\cite{Frembs2025}.} $C(S)$ thus is the unique complex-linear extension of the Boolean algebra $\mathrm{span}_\B\{\cP(S)\}$ of its projections $\cP(S)$. For our purposes it will be sufficient to work at this level, that is, we will identify contexts with Boolean algebras. The family of all such Boolean algebras (equivalently, commutative algebras) defines a partial order or \emph{poset}, under inclusion,
\begin{align*}
    \CO(=\cC(\cP(\cO)))
    :=\Big(\{C(S)\}_{S\subset\cO,[S,S]=0},\subset\Big)\; .
\end{align*}
We denote by $\cC_\mathrm{max}(\cO)$ the top layer of this poset. Importantly, we consider $\cC(\cP(\cO))$ not merely as a poset, but as the family of Boolean subalgebras together with their order relations, that is, we consider $\cC(\cP(\cO))$ as the order-preserving diagram (or covariant functor) over its underlying poset.

More generally, let $\cP$ be an event algebra \cite{Frembs2025}, that is, a partial Boolean algebra \cite{KochenSpecker1967,Greechie1968,HardingNavara2011,Kochen2015,HardingHeunenLindenhoviusNavara2019,CannonDoering2018,Frembs2025} \emph{without} the assumption that sets of mutually compatible elements are necessarily contained in a common Boolean subalgebra; the latter condition is sometimes referred to as Specker's principle \cite{Specker1960,Cabello2012,FritzEtAl2013}. We denote by $\CP$ the Boolean algebras of $\cP$, which form a poset under inclusion (which we will also denote by $\CP$). $\cP$ is the collection of Boolean algebras, $\cP=\bigcup_{C\in\CP}C$, identified (`pasted') along common Boolean subalgebras $C\cap C'$,\footnote{All Boolean algebras are assumed finite and atomic.} with least element $B_0=\B\{\one\}$, unit $\one\in\cP$ and dimension function $\mathrm{dim}:\cP\ra\mathbb{N}$ \cite{Frembs2025}. The relations $\perp$, $\vee$, $\wedge$, equivalently $+$, $\cdot$ on $\cP$ are those of its Boolean subalgebras.\footnote{A Boolean algebra defines a Boolean ring by $p\cdot q:=p\wedge q$ and $p+q:=(p\wedge q^\perp)\vee(p^\perp\wedge q)$, and vice versa via $p^\perp:=1+p$, $p\wedge q:=pq$, and $p\vee q:=p+q+pq$.} In particular, the compatibility relation $\odot$ on $\cP$ becomes $p\odot q\Leftrightarrow\exists C\in\CP$ such that $p,q\in C$.

We will refer to the poset of its Boolean algebras $\CP$ as the context poset of $\cP$, to the elements in its top level, $\CPm$, as maximal contexts and to its minimal elements as $\Pm$ (defined by $p\in\Pm$ if and only if $p\neq 0$ and for all $0\neq q\in\cP$ and $C\in\CP$ with $q,p\in C$, $q\leq p$ implies $q=p$) as atoms. Throughout, we will assume that $C\cong C'$ for all $C,C'\in\CPm$, that is, that $C,C'$ contain equally many atoms, as well as that $\cP$ admits a maximally mixed state $\rho_\one$, and thus a (finitely additive) map $\dim:\cP\ra\mathbb{N}$ with $\dim(p)=\dim(\one)\rho_\one(p)$ for all $p\in\cP$.\footnote{Clearly, this is the case for the poset of all contexts of a quantum system, denoted by $\CH$. Moreover, it holds (by construction) for the context poset of the $n$-qubit Pauli group $\ConP$ in Eq.~(\ref{eq: Pauli context poset}), as well as the context poset generated by stabiliser states (see Eq.~(\ref{eq: stab contexts})) below.} In this case, (up to rescaling) $\dim(p)=1$ for all $p\in\Pm$ and $\dim(\one):=|\Pm(C)|$ for any $C\in\CPm$.\\

\textbf{Classical embeddings.} In their seminal work, Kochen and Specker \cite{Specker1960,KochenSpecker1967} define noncontextuality as the existence of an embedding $\epsilon:\cO\ra L_\infty(X,\mu)$ of $\cO$ into the algebra of measurable functions on an underlying measurable space $X$, preserving only the algebraic relations between commuting operators. Formally, this constraint can be expressed in terms of $\epsilon$ preserving the spectral calculus, in the sense that
\begin{align}\label{eq: noncontextuality}
    \epsilon(f(O))
    =f(\epsilon(O))
\end{align}
for every $O\in\cO$ and measurable function $f:\R\ra\R$.

It is not too hard to see that Kochen-Specker (KS) noncontextuality of $\cO$ is a property of $\CP$ alone \cite{Frembs2024,Frembs2025}. This is made explicit using the concepts of context connections (see Def.~\ref{def: context connection}) and contextual holonomy (see Def.~\ref{def: contextual holonomy}). Using these tools, Kochen-Specker noncontextuality of $\cP$, that is, the existence of an injective morphism of $\widetilde{\epsilon}:\cP\ra B$ with $\widetilde{\epsilon}|_C:C\ra B$ a Boolean algebra homomorphism for every $C\in\CP$ and $B$ a single Boolean algebra \cite{Specker1960,KochenSpecker1967,VanDenBergHeunen2012,AbramskyBarbosa2021,Frembs2024,Frembs2025}, is equivalent to the existence of a context connection $\nabla$ on $\CP$ such that
\begin{align}\label{eq: flatness}
    \Hol_{C_0}(\nabla)
    =\{\id\}\; .
\end{align}
In fact, Eq.~(\ref{eq: flatness}) is independent of the base point $C_0\in\CPm$ \cite{Frembs2024,Frembs2025}, and we will thus drop it from now on.\\

\textbf{Existence of valuations.} For clarity, we contrast the above characterisation with the following notion. Kochen and Specker (and independently Bell \cite{Bell1966}) show that classical embeddings do not exist for $\LHsa$ with $\dim(H)\geq 3$, by proving the nonexistence of a necessary criterion for KS noncontextuality: the nonexistence of a valuation $v:\LHsa\ra\R$ defined to satisfy the spectral condition $v(O)\in\mathrm{spec}(O)$, and $v(OO')=v(O)v(O')$ as well as $v(O+O')=v(O)+v(O')$ for all $O,O'\in\cO$ with $[O,O']=0$. Necessity follows easily: if there exists a classical embedding $\epsilon:\LHsa\ra L_\infty(X,\mu)$, then every microstate $x\in X$  defines a valuation $v_x(O)=\epsilon(O)(x)$. However, sufficiency does not hold \cite{YuOh2012,Frembs2024}, as can be seen by comparing with the equivalent condition for KS noncontextuality in Eq.~(\ref{eq: flatness}).

\begin{corollary}\label{cor: valuation stabiliser condition}
    Let $\cO\subset\LHsa$ and $\CP\subset\CH$ its context poset. Then $\cO$ admits a valuation if and only if there exists a projection $p_0\in\Pm(C_0)$ for some $C_0\in\CPm$ and a context connection $\nabla$ on $\CP$ such that
    \begin{align}\label{eq: valuation stabiliser condition}
        \Hol(\nabla)
        \leq\Stab(p_0)\; .
    \end{align}
\end{corollary}

\begin{proof}
    This is a special case (for $p_0\in\Pm(C_0)$, and $\cO\subset\LHsa$) of the more general equivalence between noncontextual properties and frame functions in Thm.~\ref{thm: NP=2FP} below.
\end{proof}

Cor.~\ref{cor: valuation stabiliser condition} obtains the following natural interpretation. If a context connection $\nabla$ has the property in Eq.~(\ref{eq: valuation stabiliser condition}), then it is possible to label the eigenspace $p_0$ in a manner that is consistent across all contexts. In this case, we may thus call $p_0$ a \emph{noncontextual property of $\cO$}.\\

\textbf{Boolean-valued frame functions.} We next turn towards the proof of Thm.~\ref{thm: NP=2FP}, relating generalised noncontextual properties with Boolean-valued frame functions. Let $\cP$ be an event algebra with context poset $\CP$, e.g. $\cP\subset\PH$ may be a subset of projections closed under orthogonal sums and complements. Then $f:\cP\ra\R_+$ is a \emph{frame function of weight $k$} if $f(p+p')=f(p)+f(p')$ for all $p,p'\in\cP$ with $pp'=0$ and $\sum_{p\in\Pm(C)}f(p)=k$ for all $C\in\cC(\cP)$ and $\Pm(C)$ the atoms in $C$. We call $f$ \emph{Boolean-valued} if $f(p)\in\{0,1\}$ for atoms $p\in\Pm$. Here, as well as in App.~\ref{app: symplectic contextuality} we will often write $f$ for $f|_\Pm$, that is, we identify $f|_\Pm$ with its associated frame function $f:\cP\ra\zz$ defined on all elements $\cP$. To give an example, in our analysis of stabiliser theory in App.~\ref{app: classification}, minimal projections have rank $1$ and correspond with stabiliser states $\StabS\cong\Pm\subset\PHone$.

\begin{proof}[Proof of Thm.~\ref{thm: NP=2FP}]
    ``$\Leftarrow$'': Assume that $f$ is a Boolean-valued frame function. For every maximal context $C\in\CPm$ with corresponding minimal elements (atoms) $\Pm(C)$, $f$ defines the decomposition (regarding $C$ as a Boolean ring)
    \begin{align}\label{eq: f-decomposition}
        C
        =C^f_0\oplus C^f_1
        :=(p^f_C)^\perp C(p^f_C)^\perp\oplus p^f_CCp^f_C\; ,
    \end{align}
    where $p^f_C:=\sum_{\substack{p\in\Pm(C)\\f(p)=1}}p$. For any two maximal contexts $C,C'\in\CPm$, let
    \begin{align}\label{eq: frame overlaps}
        \nabla_{C'C}(p^f_C)
        =p^f_{C'}
    \end{align}
    and extend $\nabla_{C'C}$ to a Boolean isomorphism satisfying $\nabla_{C'C}|_{C\cap C'}=\mathrm{id}$. Indeed, since $f$ is additive, the elements $p^f_C$ (and corresponding context decompositions in Eq.~(\ref{eq: f-decomposition})) agree on overlaps under coarse-graining, in the sense that
    \begin{align}\label{eq: f constraint}
        |\{q\in\Pm(C)\mid q\leq p,\ q\leq p^f_C\}|
        \ =\ f(p)
        \ =\ |\{q\in\Pm(C')\mid q\leq p,\ q\leq p^f_{C'}\}|\; ,
    \end{align}
    for any atom $p\in C\cap C'$. Moreover, by the assumption that $\cP$ admits a maximally mixed state $\rho_\one$, we also have
    \begin{align}\label{eq: dim constraint}
        |\{q\in\Pm(C)\mid q\leq p\}|
        \ =\ \dim(p)
        \ =\ |\{q\in\Pm(C')\mid q\leq p\}|\; ,
    \end{align}
    for any atom $p\in C\cap C'$. Together these imply that one can identify (un)selected atoms by $f$ between the contexts $C,C'$, hence, $\nabla_{C'C}$ with $\nabla_{C'C}|_{C\cap C'}=\mathrm{id}$ exists, and we thus obtain a context connection $\nabla=(\nabla_{C'C})_{C,C'\in\CPm}$ whose elements satisfy Eq.~(\ref{eq: frame overlaps}), and since concatenation along context cycles preserves the (un)selected elements in Eq.~(\ref{eq: f-decomposition}) by construction, we have $\Hol_C(\nabla)\leq\Stab(p^f_C)$. Consequently, $f$ defines a noncontextual property of $\CP$.

    ``$\Rightarrow$'': Let $p_0\in C_0$ and let $\nabla$ be a context connection with $\Hol(\nabla)\leq\Stab(p_0)$. Define a map $f_\nabla:\Pm\ra\{0,1\}$ by
    \begin{align}
        f_\nabla(q)
        =\begin{cases}
            1&\mathrm{if}\ \nabla_{C_0C}(q)\leq p_0 \\
            0&\mathrm{otherwise}
        \end{cases}\; .
    \end{align}
    $f_\nabla(q)$ is well-defined if and only if either $\nabla_{C_0C}(q)\leq p_0$ or $\nabla_{C_0C'}(q)\not\leq p_0$ for all (maximal) contexts $C,C'$ containing $q$. This is the case since $q\in C,C'\in\CPm$ with $\nabla_{C_0C}(q)\leq p_0$ and $\nabla_{C_0C'}(q)\not\leq p_0$ is incompatible with
    \begin{align}\label{eq: holonomy condition}
        P_\gamma(\nabla)
        =\nabla_{C_0C'}\circ\nabla_{C'C}\circ\nabla_{CC_0}
        \in\Stab(p_0)\; ,
    \end{align}
    and $\nabla_{C'C}|_{C\cap C'}=\mathrm{id}$. Hence, $f_\nabla$ is well-defined on atoms. Moreover, $f_\nabla$ extends to a frame function on $\cP$ by setting 
    \begin{align*}
        f_\nabla(p)
        :=|\{q\in\Pm(C):q\leq p, \nabla_{C_0C}(q)\leq p_0\}\; ,
    \end{align*}
    which is independent of the context $C$ since for $p\in C\cap C'$, $\nabla_{C'C}$ is a Boolean isomorphism fixing $C\cap C'$ pointwise, hence, it restricts to a bijection from the atoms of $C$ below $p$ onto the atoms of $C'$ below $p$. By the same holonomy computation as in Eq.~(\ref{eq: holonomy condition}), $f_\nabla(\nabla_{C'C}(q))=f_\nabla(q)$ for every atom $q\in C$, and the two counts thus agree.
\end{proof}

In particular, the problem of characterising all noncontextual properties of a quantum system $\cO\subset\LHsa$ is therefore equivalent to characterising all Boolean-valued frame functions on $\cP=\PO$ for $\cO\subset\LHsa$. Thm.~\ref{thm: NP=2FP} thus generalises the problem of characterising sets of observables that admit no valuations.

In the rest of this paper, we study noncontextual properties of the $n$-qubit Pauli group, by classifying its Boolean-valued frame functions. Since a frame function assigns values to elements in $\cP$ independent of their context, that is, independent of the Boolean algebra they are part of, and noncontextual properties in Def.~\ref{def: noncontextual property} correspond with frame functions by Thm.~\ref{thm: NP=2FP}, the former also have to be independent of context. The proof of Thm.~\ref{thm: NP=2FP} establishes this as a general consequence of the compatibility relations under coarse-graining in a context poset. For our analysis below, it will be useful to make this explicit in the case of the Pauli group (as well as its underlying symplectic theory).

To this end, the next lemma shows that the (codespace) projectors $\Pi_{U,\chi}$ can be written as sums over orthogonal rank-$1$ projectors (corresponding to stabiliser states) in, yet independent of the choice of Lagrangian subspaces $L\supset U$.

\begin{lemma}\label{lm: character restriction}
    Let $\tU\in\Iso(V)$ be an isotropic subspace of the symplectic vector space $(V=\zz^{2n}_2,\omega)$. Then for any other isotropic subspace $U\in\Iso(V)$ with $\tU\subset U$,
    \begin{align*}
        \Pi_{\tU,\chi}
        &=\sum_{\substack{\eta\in U^*\\ \eta|_\tU=\chi}}\Pi_{U,\eta}\; .
    \end{align*}
\end{lemma}

\begin{proof}
    Using Eq.~(\ref{eq: symplectic codespace projectors}), and choosing elements $\lambda_\tU,\lambda_U$ such that $\lambda_\tU=\lambda_U|_\tU$ and with coboundary $\beta_W|_{U\times U}$, one finds
    \begin{align*}
        \sum_{\substack{\eta\in U^*\\ \eta|_\tU=\chi}}\Pi_{U,\eta}
        =\frac{1}{|U|}\sum_{v\in U}(-1)^{\lambda_U(v)}\left(\sum_{\substack{\eta\in U^*\\ \eta|_\tU=\chi}}(-1)^{\eta(v)}\right)W_v
        =\frac{1}{|U|}\sum_{v\in\tU}(-1)^{\lambda_U(v)}\left(\frac{|U|}{|\tU|}(-1)^{\chi(v)}\right)W_v
        =\Pi_{\tU,\chi}\; ,
    \end{align*}
    where we used the character orthogonality relation,
    \begin{align*}
        \sum_{\substack{\eta\in U^*\\ \eta|_\tU=0}}(-1)^{\eta(v)}
        \ =\ \begin{cases}
            |(U/\tU)^*|
            &\mathrm{if}\ v\in\tU \\
            0 &\mathrm{otherwise}
        \end{cases}\; ,
    \end{align*}
    together with the fact that the space of characters $\{\eta\in U^*\mid\eta|_\tU=\chi\}$ is an affine coset of $\{\eta\in U^*\mid\eta|_\tU=0\}\cong(U/\tU)^*$ which has cardinality $|(U/\tU)^*|=|U^*|/|\tU^*|$.
\end{proof}

For later use, we remark that Lm.~\ref{lm: character restriction} also applies to the vectors $\pi_{U,\chi}$ in Eq.~(\ref{eq: symplectic projectors}) for which $\lambda_U=0$ for all $U\in\Iso(V)$.

\section{No noncontextual properties in the $n$-qubit Pauli group}\label{app: classification}

In this section, we prove Thm.~\ref{thm: stabiliser frame functions}, that is, we rule out nontrivial Boolean-valued frame functions on $n$-qubit stabiliser states $\StabS$, even if additivity is restricted to contexts corresponding to isotropic subspaces only (cf. Eq.~(\ref{eq: Pauli context})).

To this end, we evaluate the constraint that $f$ is Boolean-valued. Let $L\in\Lag(V)$ be a Lagrangian subspace of the symplectic vector space $(V=\zz^{2n}_2,\omega)$. Define $f_L:=f|_L:L^*\ra\{0,1\}$ by $f_L(\chi):=f(\Pi_{L,\chi})$, and define $a:V\ra\zz$ by
\begin{align}\label{eq: L-restricted Fourier coefficients}
    a_v&:=(-1)^{\lambda_L(v)}b_{v,L} &
    b_{v,L}
    &:=\hf_L(v)
    =\sum_{\chi\in L^*}(-1)^{\chi(v)}f_L(\chi)\; ,
\end{align}
where $\lambda_L$ is chosen such that $\delta\lambda_L=\beta_W|_{L\times L}$ for $\beta_W=\frac{\gamma_W}{2}$ the restriction of the $2$-cocycle $\gamma_W$ of the Weyl section in Eq.~(\ref{eq: symplectic codespace projectors}) to commuting pairs. Since, by assumption, $f$ is additive on orthogonal stabiliser states $\Pi_{\psi}=\Pi_{L,\chi}$, $a$ is well-defined, that is, $a_v$ does not depend on the Lagrangian subspace $L$ containing $v\in V$.
Indeed, from Eq.~(\ref{eq: symplectic codespace projectors}),
\begin{align*}
    a_v
    \ =\sum_{\chi\in L^*}(-1)^{\chi(v)+\lambda_L(v)}f_L(\chi)\ 
    =\sum_{\substack{\chi\in L^*\\ \chi(v)+\lambda_L(v)=0}}f_L(\chi)\ -\sum_{\substack{\chi\in L^*\\ \chi(v)+\lambda_L(v)=1}}f_L(\chi)\
    =\ f(\Pi^+_{\langle v\rangle})\ -\ f(\Pi^-_{\langle v\rangle})\; ,
\end{align*}
where $\Pi^\pm_{\langle v\rangle}$ denote the projections onto the $\pm 1$-eigenspaces of $W_v$, and we used that $f$ is finitely additive. Conversely,
\begin{align}\label{eq: inverse FT}
    f_L(\chi)
    =\frac{1}{2^n}\sum_{v\in L}(-1)^{\chi(v)}b_{v,L}
    =\frac{1}{2^n}\sum_{v\in L}(-1)^{\chi(v)+\lambda_L(v)}a_v\; .
\end{align}
Now, if $f_L(\chi)$ is Boolean-valued then $f_L(\chi)=f^2_L(\chi)$, and Parseval's identity takes the form
\begin{align}\label{eq: Parseval identity}
    \frac{1}{2^n}\sum_{v\in L}a^2_v
    =\frac{1}{2^n}\sum_{v\in L}b^2_{v,L}
    =\frac{1}{2^n}\sum_{v\in L}\hf^2_L(v)
    =\sum_{\chi\in L^*}f^2_L(\chi)
    =\sum_{\chi\in L^*}f_L(\chi)
    =k\; ,
\end{align}
for every $L\in\Lag(V)$; equivalently, $\sum_{0\neq v\in L}a^2_v=k(2^n-k)$ since $a_0=b_0=k$.\\

As we will see, Eq.~(\ref{eq: Parseval identity}) implies that Boolean-valued frame functions on $\StabS$ are constant. We will split the proof into several steps, beginning with frame functions of odd weight, which in fact correspond with valuations on $\aP_n$.

\begin{lemma}\label{lm: no odd-weight frame functions}
    Let $f:\StabS\ra\{0,1\}$ for $n\geq 2$ be a frame function of weight $k$. Then $k=0\mod 2$.
\end{lemma}

\begin{proof}
    Let $f:\StabS\ra\{0,1\}$ be a frame function of odd weight $k$. For any Hermitian Pauli $P\in\aP_n$, let $\Pi^\pm_P=\frac{1}{2}(\one\pm P)$ be the projector onto its $\pm 1$-eigenspace. Since $f(\Pi^+_P)+f(\Pi^-_P)=f(\one)=k=1\mod 2$, there is a map $s:\aP_n\ra\{\pm 1\}$,
    \begin{align}
        s(P)
        =\begin{cases}
            +1&\mathrm{if}\quad f(\Pi^+_P)=1\mod 2\\
            -1&\mathrm{if}\quad f(\Pi^-_P)=1\mod 2\\
        \end{cases}\; ,
    \end{align}
    equivalently, $s(P)=(-1)^{f(\Pi^-_P)}$. We show that $s$ is multiplicative and thus a valuation. Let $\{P_i\}_i$ be a set of mutually commuting Pauli operators with $\prod_iP_i=(-1)^c\one$ for $c\in\zz_2$. For any joint eigenstate $\psi\in\StabS$, let $P_i\psi=(-1)^{s_i(\psi)}\psi$; by additivity of $f$, for any joint stabiliser eigenbasis $B$, we have $f(\Pi^-_{P_i})=\sum_{\psi\in B,s_i(\psi)=1}f(\psi)=\sum_{\psi\in B}f(\psi)s_i(\psi)$, and thus
    \begin{align*}
        \sum_if(\Pi^-_{P_i})
        =\sum_i\sum_{\psi\in B}f(\psi)s_i(\psi)
        =\sum_{\psi\in B}f(\psi)\sum_is_i(\psi)
        =\sum_{\psi\in B}f(\psi)c
        =kc=c\mod 2\; ,
    \end{align*}
    where we write $f(\psi):=f(\Pi_\psi)$. Clearly, this implies that $s$ is multiplicative, that is,
    \begin{align*}
        \prod_is(P_i)
        =\prod_i(-1)^{f(\Pi^-_{P_i})}
        =(-1)^{\sum_if(\Pi^-_{P_i})}
        =(-1)^c\; ,
    \end{align*}
    hence, $s$ defines a valuation on $\aP_n$. However, no such valuation exists for $n\geq 2$ \cite{Mermin1990,Peres1991,Mermin1993}, hence, $k=0\mod 2$.
\end{proof}

The Mermin-Peres square \cite{Mermin1990,Peres1991,Mermin1993} in Fig.~\ref{fig: the noncontextual MP-square} is the simplest two-qubit version for a proof of Lm.~\ref{lm: no odd-weight frame functions}. For other (standard) contextuality proofs (that is, ruling out valuations) in the $n$-qubit Pauli group for $n\geq 2$, see Refs.~\cite{Arkhipov2012,TrandafirLisonekCabello2022,MullerGiorgetti2025}.\\

\textbf{Frame functions of weight $2$.} Contextuality of the $n$-qubit Pauli group - in the sense of the nonexistence of valuations - thus rules out Boolean-valued frame functions of odd weight. What about frame functions of even weight? To rule out those, too, we first consider the case $k=2$. Recall that a \emph{quadratic refinement $q:V\ra\zz_2$ of $\omega$} is defined by
\begin{align}\label{eq: quadratic refinement}
    q(v)+q(w)+q(v+w)=\omega(v,w)\mod 2\; .
\end{align}
Quadratic refinements are distinguished by their \emph{Witt index $i_W$}, defined as the largest dimension of an isotropic subspace in $\cQ:=q^{-1}(0)$, which is either $n$ or $n-1$ \cite{Arf1941,Taylor1992}. Up to symplectic equivalence, these are of the form
\begin{equation}\label{eq: canonical form quadratic refinement}
\begin{aligned}
    q_W(v)&:=\sum_{i=1}^na_ib_i & &\Leftrightarrow\quad i_W(q)=n \\
    q(v)&:=\sum_{i=1}^na_ib_i+a_1+b_1 & &\Leftrightarrow\quad i_W(q)=n-1\; .
\end{aligned}
\end{equation}
Note also that since over $\zz_2$ there is only one type of quadratic term, two quadratic refinements of $\omega$ differ by a linear functional, which by nondegeneracy of the symplectic form $\omega$ is of the form $\alpha_w(v)=\omega(v,w)$ for some $w\in V$.

\begin{lemma}\label{lm: k=2}
    There exists no frame function $f:\StabS\ra\{0,1\}$ of weight $k=2$ for $n\geq 3$.
\end{lemma}

\begin{proof}
    The proof is split into two parts. The first shows that the two characters selected by a frame function of weight $k=2$ define a quadratic refinement of $\omega$ in the symplectic vector space $(V=\zz^{2n}_2,\omega)$ with Arf invariant $1$. The second proves that its level sets contain a Mermin-Peres square arrangement, hence, do not admit a valuation (cf. Lm.~\ref{lm: no odd-weight frame functions}).

    \textbf{Part 1.} Let $X_L=\{\xi_L=\chi+\lambda_L\mid \chi\in L^*,\ f(\Pi_{L,\chi})=1\}$ (with $\lambda_L$ fixed such that it has coboundary $\beta_W$, see Sec.~\ref{sec: NC properties Pauli group}) and thus $|X_L|=2$. For every $L\in\Lag(V)$, let $\xi_L,\xi_L+\rho_L\in X_L$ and note that $\rho_L\in L^*$. Define
    \begin{align}\label{eq: quadratic refinement from frame function}
        q(v):=\rho_L(v)\; ,
    \end{align}
    for all $v\in L$ and $L\in\Lag(V)$. To see that Eq.~(\ref{eq: quadratic refinement from frame function}) is independent of the Lagrangian subspace $L$, note that $\rho_L(v)=0$ if and only if the two affine characters in $X_L$ agree on $v\in L$, if and only if $(k^+_v,k^-_v)\in\{(0,2),(2,0)\}$ for $k^\pm_v:=f(\Pi^\pm_v)$ (and $k^\pm_v=1$ otherwise). Consequently, $\rho_L(v)=f(\Pi^\pm_v)\mod 2$ and since the latter expression is independent of $L\in\Lag(V)$ (as $\Pi^\pm_v$ is a coarse-grained stabiliser projector and $f$ an additive frame function), so is $q(v)=\rho_L(v)$.

    Clearly, $q$ is additive on orthogonal sums, since $q_L:=q|_L=\rho_L$ is linear in every Lagrangian subspace $L\in\Lag(V)$, hence, $q(v+w)=q(v)+q(w)$ for all $v,w\in V$ with $\omega(v,w)=0$. It follows that $q$ is either a linear functional or a quadratic refinement of $\omega$, in fact, we show that $q(v+w)=q(v)+q(w)+c\omega(v,w)$. To see this, pick a symplectic basis $\{e_1,\cdots,e_n,f_1,\cdots,f_n\}$ of $(V,\omega)$ (that is, $\omega(e_i,f_j)=\delta_{ij}$ and $\omega(e_i,e_j)=\omega(f_i,f_j)=0$). Since $q$ is additive on orthogonal sums, it is determined by $q(e_i)$, $q(f_i)$ and $q(e_i+f_i)$. Now, set $c_i=q(e_i+f_i)+q(e_i)+q(f_i)$ and compute (for $i\neq j$)
    \begin{align*}
        c_i+c_j
        &=(q(e_i+f_i)+q(e_i)+q(f_i))+(q(e_j+f_j)+q(e_j)+q(f_j))\\
        &=q(e_i+f_j)+q(e_j+f_i)+q(e_i)+q(f_i)+q(e_j)+q(f_j)\\
        &=2(q(e_i)+q(f_i)+q(e_j)+q(f_j))
        =0\mod 2\; ,
    \end{align*}
    where we used that $\omega(e_i+e_j,f_i+f_j)=\omega(e_i+f_j,e_j+f_i)=0$ together with additivity of $q$ in the second step. Let $x=\sum_{i=1}^nv_i:=\sum_{i=1}^n(a_ie_i+b_if_i)$, and note that $q(x)=\sum_{i=1}^nq(v_i)$ with $q(v_i)=a_iq(e_i)+b_iq(f_i)+ca_ib_i$. Consequently, $q(v)=\alpha(v)+cq_W(v)$, where $\alpha(v)$ is a linear functional and $q_W(v)=a\cdot b$ is the quadratic refinement in the Weyl representation in Eq.~(\ref{eq: Weyl operators}), and thus $q(v+w)=q(v)+q(w)+c\omega(v,w)$. Now, if $c=0$, then $q$ is linear and thus $q(v)=\omega(v,v_0)$ for some $v_0\in V$. It follows that $q|_L=0$ for some $L\in\Lag(V)$, equivalently, $\ker(q)$ contains a Lagrangian subspace, contradicting that $\xi_L\neq\xi_L+\rho_L$ (equivalently $\rho_L\neq 0$ for all $L\in\Lag(V)$), and thus that $f$ is a frame function of weight $2$. We conclude $c=1$, hence, $q$ defines a quadratic refinement of $\omega$ with Witt index $n-1$.

    \textbf{Part 2.} Next, we show that $f$ induces a valuation on the rank-$2$ codespace projectors $\Pi_{L_q,\xi_L|{L_q}}$, where $L_q:=\ker(q|_L)=L\cap\cQ$ for $\cQ=q^{-1}(0)$. Indeed, since $q|_{L_q}=0$, the two affine characters $\xi_L,\xi_L+q|_L\in X_L$ agree on $L_q$ for each Lagrangian $L\in\Lag(V)$, and we may thus define a map $s:\cQ\ra\zz_2$ by
    \begin{align*}
        s(v)
        :=\xi_L(v)\quad\quad v\in q^{-1}(0)\; .
    \end{align*}
    It follows from similar arguments to the ones following Eq.~(\ref{eq: quadratic refinement from frame function}) that $s$ does not depend on the choice of Lagrangian subspace $L$: namely, since $\xi_L(v)=\xi_L(v)+q|_L(v)$, the value of $s$ depends only on the choice of eigenspace $\Pi^\pm_v$ selected by $f$, which is independent of $L\in\Lag(V)$. Next, note that if $v,w\in\cQ$ with $\omega(v,w)=0$ then also $v+w\in\cQ$, since $q(v+w)=q(v)+q(w)+\omega(v,w)=0$. It follows that
    \begin{align*}
        s(v+w)
        =s(v)+s(w)+\beta_W(v,w)\; ,
    \end{align*}
    in particular, $s$ is independent of the choice of $\xi_L\in X_L$, hence, defines a valuation on $\cQ$. However, this is impossible, since $\cQ$ contains a Mermin-Peres square arrangement for $n\geq 3$. More precisely, since $q$ has Witt index $n-1$ (if $i_W(q)=n$, there would exist a Lagrangian $L\subset q^{-1}(0)$ which would imply $\rho_L=0$, contradicting $|X_L|=2$), it is up to a Clifford conjugation of the form $q(v)=\sum_{i=1}^na_ib_i+a_1+b_1$ in Eq.~(\ref{eq: canonical form quadratic refinement}). Let $\zz^{2(n-1)}_2\cong V'\subset V$ be the symplectic space on the last $n-1$ qubits, and note that $q|_{V'}$ reduces to the quadratic refinement in the Weyl representation (on the restriction to those $n-1$ qubits). In particular, the following Pauli labels are therefore all contained in $\cQ$:
    \begin{align*}
        e_{n-1} & &e_n & &e_{n-1}+e_n \\
        f_n & &f_{n-1} & &f_{n-1}+f_n \\
        e_{n-1}+f_n & &e_n+f_{n-1} & &e_{n-1}+e_n+f_{n-1}+f_n
    \end{align*}
    Clearly, this defines an operator arrangement as in the Mermin-Peres square, which admits no valuation \cite{Mermin1990,Peres1991,Mermin1993}.
\end{proof}

Clearly, if $f:\StabS\ra\{0,1\}$ is a frame function of weight $k$, then $f(\psi)=1-f(\psi)$ is a frame function of weight $2^n-k$. Consequently, Lm.~\ref{lm: k=2} also rules out frame functions of weight $2^n-2$. More generally, it will be sufficient to consider frame functions of weight $k\leq 2^{n-1}$.

Note also that Lm.~\ref{lm: k=2} does not apply to $n=2$, as there are not sufficiently many labels to construct a Mermin-Peres square arrangement in this case. In fact, noncontextual properties do exist in this case (cf. Fig.~\ref{fig: the noncontextual MP-square}).

\begin{proposition}\label{prop: noncontextual MP-square}
    Noncontextual properties of $\ConP$ for $n=2$ and $k=2$ exist, and correspond with frame functions
    \begin{equation}\label{eq: CNC frame functions}
        f_{q,s}(\Pi_{L,\chi})
        =\begin{cases}
            1&\mathrm{if}\ (\chi+\lambda_L)|_{L_q}=s|_{L_q}\\
            0&\mathrm{if}\ (\chi+\lambda_L)|_{L_q}\neq s|_{L_q}
        \end{cases}
    \end{equation}
    where $q$ is a quadratic refinement of Witt index $n-1$, $L_q:=L\cap q^{-1}(0)$ and $s:q^{-1}(0)\ra\zz_2$ arbitrary subject to $s(0)=0$.
\end{proposition}

\begin{proof}
    In this case, $q^{-1}(0)\backslash\{0\}$ corresponds to a maximal set of mutually anti-commuting Hermitian Pauli operators. Consequently, there are no compatibility relations on spectral values assigned to these operators under coarse-graining, and every assignment $s:q^{-1}(0)\ra\zz_2$ of such values thus defines a valuation on $q^{-1}(0)$.
\end{proof}

In fact, it is easy to see that Eq.~(\ref{eq: CNC frame functions}) extends to a frame function on $\cS^\mathrm{stab}_2$ (see App.~\ref{app: Gleason} below).

Notably, the frame functions in Eq.~(\ref{eq: CNC frame functions}) correspond with the CNC vertices of anti-commuting type in the literature on $\Lambda$-polytopes \cite{ZurelOkayRaussendorf2020,RaussendorfEtAl2020,IpekEtal2025}. However, CNC vertices do not correspond with frame functions for $n\geq 3$ (see below), nor are the noncontextual properties in the symplectic case (see App.~\ref{app: symplectic contextuality}) of CNC type.

Indeed, as we will see in this section, the case $n=2$ is exceptional from the perspective of stabiliser theory, yet hints at a rich structure of noncontextual properties in the underlying symplectic theory (see App.~\ref{app: symplectic contextuality}).\\

\textbf{Full classification.} To proceed, we further analyse the constraints on the Fourier coefficients $\hf_L(v)$.

\begin{lemma}\label{lm: Fourier coefficients}
    Let $f:\StabS\ra\{0,1\}$ be a frame function of weight $k$. Then $a_0=k$ and $a_v\in\{-k,-(k-2),\cdots,k-2,k\}$.
\end{lemma}

\begin{proof}
    First, note that $a_0=b_0=\sum_{\chi\in L^*}f(\chi)=\sum_{\chi\in L^*}f(\Pi_{L,\chi})=f(\one)=k$. Moreover, recall from Eq.~(\ref{eq: L-restricted Fourier coefficients}) that $b_{v,L}=\hf_L(v)=\sum_{\chi\in L^*}(-1)^{\chi(v)}f(\chi)$. Since $f$ is a Boolean-valued frame function of weight $k$, $b_{v,L}$ and thus also $a_v=(-1)^{\lambda_L(v)}b_{v,L}$ is a sum of $k$ terms of unit absolute value, hence, $|a_v|\leq k$. Finally, since $b_{v,L}$ and thus also $a_v$ is a sum of signs, it necessarily has the same parity as $k$, that is, $b_{v,L}=k\mod 2$ and similarly $a_v=k\mod 2$.
\end{proof}

We will also need the following lemma from symplectic geometry.

\begin{lemma}\label{lm: reduced stab theory}
    Let $f:\StabS\ra\{0,1\}$ be a frame function of weight $k$, and for every $0\neq v\in V$, let $\Pi^\pm_v=\frac{1}{2}(\one\pm W_v)$ be the $\pm 1$-eigenspace projectors. Then $\Pi^\pm_v\StabS\cong\cS^\mathrm{stab}_{n-1}$, and $f|_{\Pi^\pm_v\StabS}:\cS^\stab_{n-1}\ra\{0,1\}$ is a frame function of weight $f(\Pi^\pm_v)$.
\end{lemma}

\begin{proof}
    Let $(V=\zz^{2n}_2,\omega)$ be the symplectic vector space, let $0\neq v\in V$ and define its orthogonal complement by
    \begin{align*}
        v^\perp
        :=\{w\in V\mid\omega(v,w)=0\}\; .
    \end{align*}
    Since $\omega$ is alternating, $\langle v\rangle\subset v^\perp$. Moreover, since $\omega$ is nondegenerate, $\langle v\rangle=v^\perp\cap(v^\perp)^\perp$, and it follows that $V_v:=v^\perp/\langle v\rangle$ inherits a nondegenerate symplectic form $\overline{\omega}$, given by $\overline{\omega}(\ox,\oy):=\omega(x,y)$. Note that this is well-defined since
    \begin{align*}
        \omega(\tx,\ty)
        =\omega(\tx+v,\ty)
        =\omega(\tx,\ty+v)
        =\omega(\tx+v,\ty+v)
    \end{align*}
    for all $\tx,\ty\in v^\perp$. Now, $\dim(V_v)=(2n-1)-1=2(n-1)$, hence, $(V_v\cong\zz^{2(n-1)}_2,\overline{\omega})$ corresponds to the symplectic vector space of $n-1$ qubits. Moreover, for $P=W_v$, $v\neq0$ every Pauli $W_w$ with $w\in v^\perp$ preserves the eigenspace decomposition $\Pi^\pm_v$, hence, $\Pi^\pm_v\StabS=\{\psi\in\StabS\mid W_v\psi=\pm\psi\}$ are isomorphic to the stabiliser subtheory of $n-1$ qubits.

    To see that $f$ restricts to frame functions on $\Pi^{\pm}_v\StabS\cong\cS^\stab_{n-1}$, we extend the correspondence in the first part of the proof from Lagrangian to arbitrary isotropic subspaces. Isotropic subspaces of $V_v=v^\perp/\langle v\rangle$ correspond bijectively to isotropic subspaces $U\subseteq v^\perp$ with $v\in U$, via $U\mapsto \overline{U}:=U/\langle v\rangle$. Clearly, this correspondence preserves inclusion and restricts to a bijection between $\Lag(V_v)$ and the Lagrangian subspaces of $V$ containing $v$.

    Every atom $\Pi$ of $C_U$ determines a function $\xi_\Pi:U\to\zz_2$ by $W_w\Pi=(-1)^{\xi_\Pi(w)}\Pi$ for every $w\in U$, in fact $\Pi\mapsto\xi_\Pi$ is a bijection from the atoms of $C_U$ onto the affine torsor $X^\beta_U:=\{\xi:U\to\zz_2\mid \delta\xi=\beta_W|_{U\times U}\}$ over $U^*$. This identification is compatible with restriction by construction: if $\tU\subseteq U$ and $\Pi\leq\widetilde{\Pi}$ for atoms $\Pi\in C_U$, $\widetilde{\Pi}\in C_\tU$, then $\xi_{\widetilde{\Pi}}=\xi_\Pi|_\tU$.
    
    Now, let $U\in\mathrm{Iso}(V)$ with $v\in U$ and fix a sign $s\in\{\pm\}$. Since $\Pi\leq\Pi^s_v$ if and only if $(-1)^{\xi_\Pi(v)}=s$, and since evaluation at $0\neq v$ splits $X^\beta_U$ into sets of equal size, the atoms of $C_U$ split accordingly into two halves lying below $\Pi^+_v$ and $\Pi^-_v$, respectively, and each of cardinality $2^{\dim(U)-1}=|\overline{U}^*|$. Under the identification $\Pi^s_v\StabS\cong\cS^\stab_{n-1}$, the atoms of the context $C_{\overline{U}}$ of the $(n-1)$-qubit theory thus correspond precisely to those atoms of $C_U$ that lie below $\Pi^s_v$.
    
    This correspondence is compatible with coarse-graining. Indeed, let $\tU\subseteq U$ with $v\in\tU$. By Lm.~\ref{lm: character restriction}, every atom $\widetilde{\Pi}$ of $C_\tU$ is the sum of the atoms $\Pi$ of $C_U$ with $\Pi\leq\widetilde{\Pi}$, and each of these satisfies $\xi_\Pi(v)=\xi_{\widetilde{\Pi}}(v)$ since $v\in\tU$. Hence, $\widetilde{\Pi}\leq\Pi^s_v$ if and only if $\Pi\leq\Pi^s_v$ for every atom in the sum. Consequently, $f^s_v$ inherits well-definedness and finite additivity on the contexts $C_{\overline{U}}$ of the $(n-1)$-qubit theory from that of $f$ on the contexts $C_U$ with $v\in U$, restricted to the atoms below $\Pi^s_v$. Clearly, $f^s_v$ is Boolean-valued. Finally, taking $\tU=\langle v\rangle$ (equivalently $\overline{\tU}=\{0\}$) identifies the
    weight as,
    \begin{align}
        \sum_{\chi\in\oL^{*}} f^s_v(\Pi_{\oL,\chi})
        =f(\Pi^s_v)
        =:k^s_v\; ,
    \end{align}
    for all $\oL\in\Lag(V_v)$, hence, $f^s_v$ is a Boolean-valued frame function of weight $k^s_v$, as claimed.
\end{proof}

Note that the proof uses only the contexts $C_U$ with $U\in\Iso(V)$, and thus applies verbatim to frame functions defined with respect to $\ConP$ rather than $\ConS$. It also applies verbatim in the symplectic setting, upon replacing $\Pi_{U,\chi}$ by $\pi_{U,\chi}$, $\Pi^{\pm}_v$ by $\pi^{\pm}_v=\tfrac{1}{2}(e_0\pm e_v)$, $\cS^\stab_n$ by $\cS^{\mathrm{symp}}_n$ and setting $\lambda_U\equiv 0$ (see App.~\ref{app: symplectic contextuality} below).

In light of the existence of non-maximal unextendible sets of orthogonal stabiliser states (see Ref.~\cite{Frembs2026_unextendibility}), we remark that the sets of $2^{n-1}$ orthogonal states corresponding to a basis in $\cS^\stab_{n-1}$ of Lm.~\ref{lm: reduced stab theory} constitute a special case for which extendibility does hold: a complete reduced stabiliser basis lifts to a complete stabiliser basis of either eigenspace $\Pi_v^\pm\cH$, and adjoining the corresponding lift in the opposite eigenspace yields a complete $n$-qubit stabiliser basis.

Before considering the remaining case for $n=3$ (after Lm.~\ref{lm: no odd-weight frame functions} and Lm.~\ref{lm: k=2}), we briefly recall the following. The symplectic polar space $W(2n-1,2)$ is the point–line geometry whose points are the projective points $\langle v\rangle$, $0\neq v\in V$, and whose lines are the totally isotropic triples $\{v,w,v+w\}$ with $\omega(v,w)=0$. Its generators are the Lagrangian subspaces. For $m\geq 1$, an $m$-ovoid of $W(2n-1,2)$ is a set of points meeting every generator in exactly $m$ points.

\begin{lemma}\label{lm: k=4}
    There exists no stabiliser frame function $f:\StabS\ra\{0,1\}$ of weight $k=4$ for $n\geq 3$.
\end{lemma}

\begin{proof}
    Consider first the case $n=3$, and let $f:\StabS\ra\{0,1\}$ be a frame function. By Lm.~\ref{lm: Fourier coefficients}, $a_v\in\{-4,-2,0,2,4\}$ for all $v\in V$. Next, fix any non-zero Pauli operator $\one\neq W_v\in\pm\aP_n$ and consider the projectors onto its eigenspaces $\Pi^\pm_v=\frac{1}{2}(\one\pm W_v)$. Let $k^\pm_v=f(\Pi^\pm_v)$ be their weights and note that $k^+_v+k^-_v=k$. Now, by Lm.~\ref{lm: reduced stab theory}, $\Pi^\pm_v\cS^\mathrm{stab}_3\cong\cS^\mathrm{stab}_2$ and $f^\pm_v:\cS^\mathrm{stab}_2\ra\{0,1\}$ reduces to a Boolean-valued frame function. By Lm.~\ref{lm: no odd-weight frame functions}, this implies that $k^\pm_v\in\{0,2,4\}$, hence, $a_v=k^+_v-k^-_v\in\{-4,0,4\}$. Consequently, Parseval's identity, Eq.~(\ref{eq: Parseval identity}) reads $\sum_{0\neq v\in L}a^2_v=k(2^n-k)=16$ for every $L\in\Lag(V)$, hence, implies that there exists a single non-zero element $a_v$ ($v\neq 0$) in every Lagrangian subspace. In other words, the Fourier coefficients $\{a_v\}_{0\neq v\in V}$ thus define a $1$-ovoid in the symplectic polar space $W(5,2)$, where $W(2n-1,2)$ is defined as the incidence geometry whose points are the projective points $\langle v\rangle$ for every $0\neq v\in V$, and whose lines correspond with projective isotropic lines $\{v,w,v+w\}$. However, it is well-known that no $1$-ovoids exist (for $n\geq 3$), hence, neither do Boolean-valued frame functions of weight $4$.\footnote{Note that a $1$-ovoid translates into a set of mutually anti-commuting Hermitian Pauli operators, whose cardinality is that of a partition of $\aP_n$ into disjoint maximal Abelian subgroups, that is, $2^n+1$. Yet, maximal sets of mutually anti-commuting Pauli operators have cardinality $2n+1$ (see Ref.~\cite{SarkarVanDenBerg2021}, Lm.~8), hence, $1$-ovoids exist for $n\leq 2$ only.}
    
    The extension to $n\geq 3$ now follows by induction with Lm.~\ref{lm: reduced stab theory}. Assume that there exist no Boolean-valued frame functions for $n-1$ qubits of weight $4$, and let $f:\StabS\ra\{0,1\}$ be a frame function of weight $k=4$ for $n$ qubits. Let $0\neq v\in V$, and consider the reduced frame function on the stabiliser sub-theories on $\Pi^\pm_v$ with weight $k^+_v+k^-_v=4$. Since, by Lm.~\ref{lm: no odd-weight frame functions}, $k^\pm_v=0\mod 2$ and $k^\pm_v\neq 4$ we must have $k^\pm_v=2$ and thus $a_v=f(\Pi^+_{\langle v\rangle})-f(\Pi^-_{\langle v\rangle})=k^+_v-k^-_v=0$ for all $0\neq v\in V$. Yet, since $a_v=(-1)^{\lambda_L(v)}b_{v,L}$ this contradicts $f_L(\chi)=\frac{1}{2^n}\sum_{v\in L}(-1)^{\chi(v)}b_{v,L}=\frac{4}{2^n}$ being Boolean-valued (for $n\geq 3$).
\end{proof}

Finally, an induction argument rules out nontrivial Boolean-valued stabiliser frame functions for all $n\geq 3$.

\begin{proof}[Proof of Thm.~\ref{thm: stabiliser frame functions}]
    By Lm.~\ref{lm: no odd-weight frame functions}, Lm.~\ref{lm: k=2} and Lm.~\ref{lm: k=4}, every frame function with $n=3$ is constant.

    Now, assume that $f:\StabS\ra\{0,1\}$ is a frame function for $n>3$, and with Lm.~\ref{lm: reduced stab theory} consider its restriction to the stabiliser subtheory $\Pi^\pm_v\StabS\cong\cS^\mathrm{stab}_{n-1}$ for $0\neq v\in V$. By the inductive hypothesis, $f^\pm_v$ is constant, that is, $k^\pm_v=f(\Pi^\pm_v)\in\{0,2^{n-1}\}$ and thus $k\in\{0,2^{n-1},2^n\}$. The only nonconstant frame function among those has weight $2^{n-1}$. However, such a frame function selects, for every non-zero Pauli label $0\neq v\in V$ one of its eigenspaces, hence, defines a map
    \begin{align*}
        s(W_v):=
        \begin{cases}
            +1&\mathrm{if}\ f(\Pi^+_v)=2^{n-1}\\
            -1&\mathrm{if}\ f(\Pi^-_v)=2^{n-1}
        \end{cases}\; ,
    \end{align*}
    which is a valuation, since for any two commuting Pauli operators $W_v,W_w$ with $\omega(v,w)=0$ (yet, possibly $W_vW_w=\epsilon W_{v+w}$ for $\epsilon=\pm 1$), the $s$-selected eigenspaces of $W_v$ and $W_w$ intersect nontrivially and thus $s(W_vW_w)=s(W_v)s(W_w)$, by additivity of $f$. Yet, for $n\geq 2$ no valuation exists, hence, $k\in\{0,2^n\}$, that is, $f:\StabS\ra\{0,1\}$ is constant.
\end{proof}

We remark that Thm.~\ref{thm: stabiliser frame functions} already holds on the level of the Pauli context poset $\ConP$, which contains far fewer contexts than the stabiliser poset $\ConS$, and thus puts far fewer constraints on a frame function defined over it; that is, it already holds by restricting additivity of stabiliser frame functions $f:\StabS\ra\{0,1\}$ to contexts $C\in\ConP$.

\section{Quaternionic Valuations}\label{app: quaternionic valuations}

In this section, we relate the nonexistence of noncontextual properties in the $n$-qubit Pauli group to the non-existence of valuations in quaternionic quantum mechanics. The main idea is the following embedding result.

\begin{lemma}\label{lm: restriction}
    Let $\cP,\cP'$ be two event algebras with corresponding context posets $\cC(\cP),\cC(\cP')$. If $i:\cP'\ra\cP$ is an inclusion such that $i(C'\cap C'')=i(C')\cap i(C'')$ for all $C',C''\in\cC(\cP')$, and $i(\cC_\mathrm{max}(\cP'))\subset\cC_\mathrm{max}(\cP)$, then every noncontextual property of $\cC(\cP)$ defines a noncontextual property on $\cC(\cP')$ by restriction.
\end{lemma}

\begin{proof}
    Since maximal contexts in $\cC(\cP')$ are maximal contexts in $\CP$, and $i(C'\cap C'')=i(C')\cap i(C'')$, the restriction of a noncontextual property $\{p_C\}_{C\in\CPm}$ to $\cC'$ is defined and inherits that $\dim(p_C)=\mathrm{const}$ for all $C\in\CPm$. Moreover, the holonomy condition in Def.~\ref{def: noncontextual property} is inherited under restriction, since every context cycle in $\cC(\cP')$ is a context cycle in $\CP$.
\end{proof}

The assumption $i(\cC'_\mathrm{max})\subset\cC_\mathrm{max}$ in Lm.~\ref{lm: restriction} cannot be dropped: if $\cC$ is a refinement of $\cC'$, equivalently if maximal contexts in $\cC'$ are coarse-grainings of maximal contexts in $\cC$, then the latter may also contain noncontextual properties that do not restrict to $\cC'$. For the cases of interest to us, $\ConP$, $\ConS$ and $\ConQ$, this assumption is satisfied.

As a consequence of Lm.~\ref{lm: restriction}, Thm.~\ref{thm: stabiliser frame functions} rules out noncontextual properties on context posets that include $\ConP$, which in particular implies the following special case of Gleason's theorem \cite{Gleason1957} already from a finite set of contexts.

\begin{corollary}\label{cor: complex case}
    Let $\cH=\C^{2^n}$ with $n\geq 3$. Then $\CH$ admits no nontrivial noncontextual properties.
\end{corollary}

\begin{proof}
    Assume that $\CH$ admits a noncontextual property, equivalently a Boolean-valued frame function on $\PHone$. By restriction, this would define a Boolean-valued frame function on $\StabS$, yet, these are ruled out by Thm.~\ref{thm: stabiliser frame functions}.
\end{proof}

Next, we consider quaternionic Hilbert spaces. Consider the $n$-qubit Hilbert space with context category $\CH$ for $\cH=\C^{2^n}$, and let $\Theta:\cH\ra\cH$ be a quaternionic structure, that is, an anti-unitary operator with $\Theta^2=-\one$. Equipped with $\Theta$, we may view $\cH$ as a quaternionic Hilbert space $\cH\cong\Q^{2^{n-1}}$. Now, consider the subposet $\ConQ\subset\CH$ generated by all projections of the form
\begin{align}\label{eq: Kramers projection}
    \Pi_\Theta
    =\Pi+\Theta\Pi\Theta^{-1}\; ,
\end{align}
where $\Pi\in\cP_1(\C^{2^n})$.\footnote{From a physical point of view, $\Pi_\Theta$ may be interpreted as the projection onto a Kramers pair.} Recall that every anti-unitary operator $A$ can be decomposed as $A=UK$, where $K$ denotes complex conjugation and $U$ is unitary. For quaternionic structures this implies $\Theta^2=UKUK=U\overline{U}=-\one$, which is satisfied, in particular, for all Hermitian Pauli operators satisfying $W^T_v=(-1)^{q_W(v)}W_v=-W_v$, that is, Pauli strings with an odd number of $Y$ factors. Write $\Theta_v:=W_vK$, and note that the action of $\Theta_v$ on Pauli operators is given by
\begin{align*}
    \Theta_vW_w\Theta^\dagger_v
    =W_vKW_wKW_v
    =W_v\overline{W}_wW_v
    =W_vW^T_wW_v
    =(-1)^{\omega(v,w)+q_W(w)}W_w
    =(-1)^{q_v(w)}W_w\; ,
\end{align*}
where $q_v(w):=q_W(w)+\omega(v,w)$ is a quadratic refinement, which has Witt index $n-1$ if and only if $q_W(v)=1$, e.g. for $W_v=Y_1$, $q_v$ is of the canonical form in Eq.~(\ref{eq: canonical form quadratic refinement}). A choice of such quadratic refinement $\Theta_{q_v}:=\Theta_v$ thus decomposes the contexts in $\ConP$ into rank-$2$ projections as in Eq.~(\ref{eq: Kramers projection}), hence, it defines a coarse-graining,
\begin{align}\label{eq: quaternionic coarse-graining}
    C_U\stackrel{q_v}{\mapsto}C_{U_{q_v}}
    &:=\mathrm{span}_\B\{\Pi_{U_{q_v},\eta}\mid\eta\in U^*_{q_v}\}\; , &
    \mathrm{span}_\C(C_{U_{q_v}})
    &=\{W_w\in C_U\mid w\in q^{-1}_v(0)\}''\; .
\end{align}
where $U_{q_v}:=U\cap q^{-1}_v(0)$. Denote the partial order of all such contexts by $\cC_{q_v}(\aP_n)$. It follows immediately that this poset corresponds with the Pauli-restricted poset $\ConQ|_{\aP_n,\Theta_{q_v}}$ for a Pauli-compatible quaternionic structure $\Theta_{q_v}$,
\begin{align}\label{eq: Pauli-compatible quaternionic structure}
    \Theta_{q_v} W_w\Theta^{-1}_{q_v}=(-1)^{q_v(w)}W_w\quad\forall w\in V\; ,
\end{align}
with $q_v$ a quadratic refinement of Witt index $n-1$.

\begin{lemma}\label{lm: quaternionic Pauli contexts}
    Let $q$ be a quadratic refinement of Witt index $n-1$, then $\ConqP\cong\ConQ|_{\aP_n,\Theta_q}$.
\end{lemma}

\begin{proof}
    By Eq.~(\ref{eq: Pauli-compatible quaternionic structure}), a Pauli operator $W_w$ is quaternionic if and only if $q(w)=0$, hence, the Pauli-restricted quaternionic contexts, which are generated by projections in Eq.~(\ref{eq: Kramers projection}), are precisely those corresponding to codespace projectors with respect to the isotropic subspaces $U\subset q^{-1}(0)$. Since, by Eq.~(\ref{eq: quaternionic coarse-graining}), $\ConqP$ is generated by the operators (labelled by elements) in these subspaces, the two context posets are isomorphic.
\end{proof}

By contrast, the projectors in general stabiliser contexts $C\in\cC_\mathrm{max}(\cP^\mathrm{stab}_n)$ cannot be partitioned into rank-$2$ projectors as in Eq.~(\ref{eq: Kramers projection}). Together with Lm.~\ref{lm: restriction} and Lm.~\ref{lm: k=2} this rules out valuations in quaternionic quantum mechanics.

\begin{proof}[Proof of Cor.~\ref{cor: no quaternionic valuations}]
    Since different quaternionic structures are unitarily related,  the posets $\ConQ$ for different quaternionic structures $\Theta$ in $\C^{2^n}$ are isomorphic, and it thus is sufficient to consider the case of a Pauli-compatible quaternionic structure $\Theta_q$ as in Eq.~(\ref{eq: Pauli-compatible quaternionic structure}). By Lm.~\ref{lm: quaternionic Pauli contexts}, $\ConqP\subset\ConQ|_{\Theta_q}\cong\ConQ$ is an inclusion and the respective maximal contexts are isomorphic (since $\dim(L_q)=n-1$, $C_{L_q}$ contains $2^{n-1}$ minimal projections of complex rank $2$, equivalently quaternionic rank $1$, which thus form a complete basis of $\Q^{2^{n-1}}$). Hence, Lm.~\ref{lm: restriction} implies that valuations on $\ConQ$ restrict to valuations on $\ConqP$. Yet, by part 2 of the proof of Lm.~\ref{lm: k=2}, no such valuations exist for $n\geq 3$.
\end{proof}

We close by commenting on this result. A quaternionic structure defines a pairing of rank-$1$ projections into Kramers pairs via Eq.~(\ref{eq: Kramers projection}). Viewed from the perspective of the ambient complex Hilbert space $\cH=\C^{2^n}$, $\ConQ$ thus defines a consistent pairing of rank-$1$ projectors in every maximal context of $\cC(\C^{2^n})$. This is analogous to the situation of $\ConP$, specifically part 1 in the proof of Lm.~\ref{lm: k=2} and may be viewed as a type of noncontextual coarse-graining.

However, comparing with the proof of Lm.~\ref{lm: k=2}, this is not sufficient to define a noncontextual property, as the latter also requires to pick a character in every context that is compatible under restriction. Yet, part 2 of the proof of Lm.~\ref{lm: k=2} asserts that, as a consequence of the nontrivial Pauli cocycle $\gamma_W$ (which restricts to $\beta_W=\frac{\gamma_W}{2}$ on commuting pairs), such a choice of characters does not exist for $n>2$. Of course, it is well-known that, as a consequence of Gleason's theorem \cite{Gleason1957,Varadarajan1985,MorettiOppio2018}, valuations (and, more generally, nontrivial Boolean-valued frame functions of any weight) do not exist in real, complex and quaternionic quantum mechanics in dimension greater than two \cite{Gleason1957,Varadarajan1985,MorettiOppio2018}. Our reduction to the $n$-qubit Pauli group may be read as a discrete argument to this end.

Finally, note that the case $n=2$ is excluded, as it corresponds to the exceptional case of Gleason's theorem for quaternionic Hilbert spaces. In light of its restriction to Pauli operators and stabiliser projectors, this can be seen explicitly from the pairing in Eq.~(\ref{eq: Kramers projection}): there are no nontrivial overlaps of the one-dimensional isotropic subspaces $L_q$, hence, no constraints on the choice of characters $\chi_{L_q}$.

\section{Classification of noncontextual properties in symplectic polar spaces}\label{app: symplectic contextuality}

In this section, we characterise noncontextual properties of the projective Pauli group, equivalently of its underlying symplectic theory. Our characterisation follows the same line of reasoning as that for Boolean-valued stabiliser frame functions in App.~\ref{app: classification}. The difference is the additional Weyl cocycle condition and the related affine characters $\lambda_U:U\ra\zz_2$ in Eq.~(\ref{eq: symplectic codespace projectors}), which are not present (identically zero) in the purely symplectic case. The relevant poset is therefore $(\cA(\cP_n/Z(\cP_n),\leq))$, where $\cA(\cP_n/Z(\cP_n)):=\{A/Z(\cP_n)\mid Z(\cP_n)<A<\cP_n,A\ \mathrm{Abelian}\}$ denotes the partial order of Abelian subgroups up to phases. Note that $(\Iso(V),\leq)\cong (\cA(\cP_n/Z(\cP_n),\leq))$ as bare posets.

For better comparison with the Pauli case, we represent this poset explicitly. Let $\R[V]$ be the Bloch space with orthogonal basis $\{e_v\}_{v\in V}$ and normalisation $\langle e_v,e_w\rangle=2^n\delta_{vw}$.\footnote{We choose the normalisation of the $\{e_v\}_{v\in V}$ to match that of the Weyl representation in Eq.~(\ref{eq: Weyl operators}).} The mapping $U\mapsto\big\{\sum_{\chi\in U^*}\lambda_\chi\pi_{U,\chi}\mid \lambda_\chi\in\B\big\}$ for every isotropic subspace $U\in\Iso(V)$ and with $\pi_{U,\chi}$ defined in Eq.~(\ref{eq: symplectic projectors}) defines a partial Boolean algebra $\cP(V)$ with corresponding context poset (under inclusion) $\ConPP$. It follows from Lm.~\ref{lm: character restriction} that $(\Iso(V),\subseteq)\cong\ConPP$ as bare posets. As a family of ordered Boolean algebras, $\ConPP$ is distinguished from $\ConP$ by the vanishing of the cohomology class of $\gamma_W$ (equivalently of $\beta_W$ for $n>1$), parametrised by  $\{\lambda_U\}_{U\in\Iso(V)}$ in Eq.~(\ref{eq: symplectic codespace projectors}) and Eq.~(\ref{eq: symplectic projectors}), for which we may w.l.o.g. choose the trivial representative. Moreover, in analogy with Eq.~(\ref{eq: stabiliser projectors}), we define
\begin{align}
    \StabSS
    :=\{\pi_{L,\chi}\mid L\in\Lag(V),\chi\in L^*\}\; .
\end{align}
Similarly to stabiliser states, these vectors generate an event algebra,
\begin{align}\label{eq: PV}
     \pV
     :=\bigcup_{C\subset\ConSS}C 
\end{align} 
where $\ConSS$ is the downward-closed set of all maximal contexts of the form
\begin{align}\label{eq: symplectic stab contexts}
    \mathrm{span}_\B\{\{\pi_{L_i,\chi_i}\}_{i=1}^{2^n}\mid\pi_{L_i,\chi_i}\in\StabSS,\ \pi_{L_i,\chi_i}\perp\pi_{L_j,\chi_j}\ \mathrm{if}\ i\neq j\}\; .
\end{align}
Here, orthogonality is defined in Eq.~(\ref{eq: stabiliser orthogonality}), which replaces the analogous orthogonality relations of stabiliser states in $\ConS$ \cite{DehaeneDeMoor2003,AaronsonGottesman2004} (with the only difference that the stabiliser states also depend on the nontrivial Weyl cocycle $\gamma_W$ via $\beta_W$). In particular, $\pi_{U,\eta}=\sum_{\substack{\chi\in L^*\\ \chi|_U=\eta}}\pi_{L,\chi}$ for the vectors represented in Eq.~(\ref{eq: symplectic projectors}), which is well-defined by Lm.~\ref{lm: character restriction}, and $\dim(\pi_{L,\chi})=1$ defines a dimension function on $\pV$ corresponding to a maximally mixed state $\rho_\one(p)=\frac{\dim(p)}{\dim(\one)}$. As with $\cP^\mathrm{stab}_n$, $\pV$ is not a partial Boolean algebra for $n\geq 3$ as there exist sets of mutually orthogonal vectors $\{\pi_{L_i,\chi_i}\}_{i=1}^S$ with $S<2^n$ yet with no other vector $\pi_{L,\chi}$ in their mutual complement (see Ref.~\cite{Frembs2026_unextendibility}).

By analogy with the projectors in Eq.~(\ref{eq: symplectic codespace projectors}), we may think of $\pi_{U,\eta}$ as `abstract symplectic projectors', and we thus call $\ConSS$ the \emph{symplectic stabiliser context poset}. In particular, note that $(\pV)_\mathrm{min}=\StabSS$ and we have $(\Iso(V),\leq)\cong\ConPP\subset\ConSS$, as $\ConSS$ contains more contexts. In turn, frame functions on $\StabSS$ are more constrained if finite additivity is understood with respect to $\ConSS$ as opposed to $\ConPP$.\\

With this notation at hand, and adhering to Thm.~\ref{thm: NP=2FP}, we proceed to characterise noncontextual properties in terms of Boolean-valued frame functions $f:\StabSS\ra\{0,1\}$. Unless stated otherwise, these are understood with respect to $\ConPP$, that is, finite additivity is imposed only on the contexts $c_U$ with $U\in\Iso(V)$; this requirement is weaker, but we will see in the proof of Thm.~\ref{thm: symplectic Boolean frame functions} that it selects the same frame functions as additivity with respect to $\ConSS$.\\

\textbf{Valuations.} Comparing with the signed Pauli case, we first note that Lm.~\ref{lm: no odd-weight frame functions} does not carry over to the symplectic case, since linear functionals $\alpha:V\ra\zz_2$ define valuations on $V$, and thus frame functions on $\pV$, of the form\footnote{There are many such linear functionals, indeed, $|V^*|=2^{2n}$. Note that this is different to a valuation, which demands linearity with respect to the projective representation of $V$ in terms of Pauli operators (see Lm.~\ref{lm: different posets}).}
\begin{align*}
    f_\alpha(\pi_{L,\chi})
    &=\begin{cases}
        1&\mathrm{if}\ \chi=\alpha|_L\\
        0&\mathrm{if}\ \chi\neq\alpha|_L
    \end{cases}\; .
\end{align*}
Yet, it is not only linear functionals that define families of characters that are compatible on overlaps. By the defining relation, $q(v)+q(w)+q(v+w)=\omega(v,w) \mod 2$ in Eq.~(\ref{eq: quadratic refinement}), every quadratic refinement $q$ provides another such family $\{q|_L\}_{L\in\Lag(V)}$. Together, these exhaust all such families.

\begin{lemma}\label{lm: compatible characters}
    Every family of characters $\{\chi_U\in U^*\}_{U\in\Iso(V)}$ with $\chi_U|_{\tU}=\chi_{\tU}$ whenever $\tU\subset U$ is induced by a unique map $\alpha+cq_W:V\ra\zz_2$ for $\alpha$ a linear functional, $c\in\zz_2$ and $q_W$ the quadratic refinement of the Weyl section in Eq.~(\ref{eq: Weyl operators}).
\end{lemma}

\begin{proof}
    Given a compatible family, define $g:V\ra\zz_2$ by $g(0)=0$ and $g(v)=\chi_{\langle v\rangle}(v)$ for all $0\neq v\in V$. Compatibility then implies $g(v+w)=g(v)+g(w)$ for all $v,w\in V$ with $\omega(v,w)=0$. By the same arguments as in part 1 of the proof of Lm.~\ref{lm: k=2}, it follows that $g(v)=\sum_{i=1}^n(a_ig(e_i)+b_ig(f_i)+ca_ib_i)$ for $c\in\zz_2$ and $\{e_i,f_j\}_{i,j=1}^n$ a symplectic basis with $v=\sum_{i=1}^n(a_ie_i+b_if_i)$. Defining the linear functional $\alpha$ by $\alpha(e_i)=g(e_i)$ and $\alpha(f_i)=g(f_i)$, we thus have $g=\alpha+cq_W$.
\end{proof}

What is more, all such functions define frame functions on $\StabSS$. Indeed, let $f_g(\pi):=\langle t_g,\pi\rangle$, where
\begin{align}\label{eq: weight-1 frame function operator}
    t_g
    :=\frac{1}{2^n}\sum_{\substack{v\in V}}(-1)^{g(v)}e_v\; ,
\end{align}
and $g=\alpha+cq_W$ as in Lm.~\ref{lm: compatible characters}. Note first that it follows with Eq.~(\ref{eq: symplectic projectors}) that
\begin{align*}
    \langle t_g,\pi_{L,\chi}\rangle
    =\frac{1}{2^{2n}}\sum_{v\in V}
    \sum_{w\in L}(-1)^{g(v)+\chi(w)}\langle e_v,e_w\rangle
    =\frac{1}{2^n}\sum_{w\in L}(-1)^{g(w)+\chi(w)}
    =\begin{cases}
        1&\mathrm{if}\ g|_L=\chi \\
        0&\mathrm{otherwise}
    \end{cases}
    =\delta_{g|_L,\chi}\; ,
\end{align*}
hence, $f_g$ defines a frame function that is additive under the restriction to contexts of $\ConPP\cong(\Iso(V),\subseteq)$.

What is more, for a maximal set of $2^n$ mutually orthogonal vectors $\{\pi_{L_i,\chi_i}\}_{i=1}^{2^n}$,
\begin{align*}
    \sum_{i=1}^{2^n}f_g(\pi_{L_i,\chi_i})
    =\langle t_g,\sum_{i=1}^{2^n}\pi_{L_i,\chi_i}\rangle
    =\langle t_g,e_0\rangle
    =1\; ,
\end{align*}
since from $\langle\sum_{i=1}^{2^n}\pi_i,\sum_{i=1}^{2^n}\pi_i\rangle=2^n$, $\langle\sum_{i=1}^{2^n}\pi_i,e_0\rangle=2^n$ and $\langle e_0,e_0\rangle=2^n$, equality in the general Cauchy-Schwarz inequality implies that $\sum_{i=1}^{2^n}\pi_i=e_0$. Consequently, $f_g$ further extends to frame function of weight $1$ on $\StabSS$.

\begin{lemma}\label{lm: different posets}
    $\ConP\cong\ConPP$ as diagrams over $(\Iso(V),\subseteq)$ if and only if $\aP_n$ admits a valuation, if and only if $n=1$.
\end{lemma}

\begin{proof}
    Given a valuation $s:\aP_n\cong V\to\zz_2$ with $s(v)+s(w)+s(v+w)=\beta_W(v,w)$ for all $v,w\in V$, $\omega(v,w)=0$, we may use the freedom in the choice of local trivialisations to define $\lambda_U:=s|_U$ consistently for all $U\in\Iso(V)$. Indeed, $\delta\lambda_U=\delta s|_U=\beta_W|_{U\times U}$, and these choices are compatible under restriction, $s|_\tU=(s|_U)|_\tU$ for $\tU\subseteq U$. Consequently, $\phi_U(\Pi_{U,\chi}):=\pi_{U,\chi}$ defines an isomorphism of Boolean algebras for every $U$, and since the family $\{\phi_U\}_{U\in\Iso(V)}$ is compatible with the inclusions of contexts, it defines an isomorphism $C(\widetilde P_n)\cong C(\overline P_n)$ as diagrams over $(\Iso(V),\subseteq)$.
    
    Conversely, a diagram isomorphism restricts on $U=\langle v\rangle$ to a bijection of two-element sets of atoms, hence, to a sign $s(v)\in\zz_2$. Compatibility with restriction from every Lagrangian $L\in\Lag(V)$ forces $\lambda_L+s|_L\in L^*$, that is, $\delta(s|_L)=\beta_W|_{L\times L}$ for all $L$. Hence, $s$ defines a valuation on $\aP_n\cong V$, and the claim thus follows since valuations on $\aP_n$ exist only for $n=1$ \cite{Mermin1990,Peres1991,Mermin1993}.
\end{proof}

\textbf{Frame functions of weight $2$.} Next, we show that the pairing of quadratic refinements $q$ of $\omega$ with Witt index $n-1$ and linear functionals also gives rise to generalised noncontextual properties of both $\ConPP$ and $\ConSS$.

\begin{lemma}\label{lm: noncontextual symplectic properties}
    Let $(V=\zz^{2n}_2,\omega)$ be a symplectic vector space, let $q$ be a quadratic refinement of $\omega$ with Witt index $n-1$, and let $\alpha:V\ra\zz_2$ be a linear functional. Then the collection of elements
    \begin{align}
        p_U
        &:=\pi_{U_q,\alpha|_{U_q}}\; ,
    \end{align}
    with $U_q:=U\cap q^{-1}(0)$ for all $U\in\Iso(V)$ defines a noncontextual property for $\ConPP\cong(\Iso(V),\leq)$.
\end{lemma}

\begin{proof}
    Let $q$ be a quadratic refinement of $\omega$ with Witt index $n-1$. This implies that $L_q:=L\cap q^{-1}(0)$ is an isotropic subspace of dimension $n-1$ in every Lagrangian subspace $L\in\Lag(V)$. Moreover, for any isotropic subspace $U\subset L$,
    \begin{align*}
        U_q
        :=U\cap q^{-1}(0)
        =(L\cap q^{-1}(0))\cap U
        =L_q\cap U\; ,
    \end{align*}
    which again is an isotropic subspace and compatible under restriction. In order to select elements from these sets, let $\alpha:V\ra\zz_2$ be a linear functional. Every such functional restricts to a character $\alpha|_U\in U^*$ in every isotropic subspace, and is compatible under inclusion. Hence, setting $p_U:=\pi_{U_q,\alpha|_{U_q}}=\sum_{\substack{\eta\in U^*\\\eta|_{U_q}=\alpha|_{U_q}}}\pi_{U,\eta}$ defines a collection of abstract projectors, which decompose into two atoms $\pi_{U,\alpha|_U},\pi_{U,(\alpha+q)|_U}$ whenever $q|_U\neq 0$, and are compatible under restriction,
    \begin{align*}
        p_U|_{C_\tU}
        =\pi_{U_q,\alpha|_{U_q}}|_{C_\tU}
        =\left(\sum_{\substack{\eta\in U^*\\\eta|_{U_q}=\alpha|_{U_q}}}\pi_{U,\eta}\right)\Bigg|_{C_\tU}
        =\sum_{\substack{\eta\in \tU^*\\\eta|_{\tU_q}=\alpha|_{\tU_q}}}\pi_{U,\eta}
        =\pi_{\tU_q,\alpha|_{\tU_q}}
        =p_\tU\; ,
    \end{align*}
    for all $C_U,C_\tU\in\ConPP\cong(\Iso(V),\leq)$ with $C_\tU\subset C_U$, where we used that, by Lm.~\ref{lm: character restriction}, $p_U|_{C_\tU}$ is the sum of elements in $C_U$ that are smaller than $p_\tU$, which are again expressible as in Eq.~(\ref{eq: symplectic projectors}), and where the joint eigenvalues of all Pauli labels in $U$ agrees with that of the actual selected element $\pi_{\tU,\alpha_{\tU}}$ in $C_\tU$.
\end{proof}

Moreover, the existence of noncontextual properties of the type in Lm.~\ref{lm: noncontextual symplectic properties} extends to the symplectic stabiliser poset $\ConSS$, which includes contexts generated by abstract stabiliser states with respect to the orthogonality relation in Eq.~(\ref{eq: stabiliser orthogonality}). In other words, they correspond with Boolean-valued frame functions $f:\StabSS\ra\{0,1\}$.

\begin{lemma}\label{lm: noncontextual property in symplectic stabiliser poset}
    Let $q$ be a quadratic refinement of $\omega$ with Witt index $n-1$, and $\alpha:V\ra\zz_2$ a linear functional. Then
    \begin{align}
        p_C
        &:=\sum_{\substack{\pi_{L,\chi}\in \Pm(C)\\\pi_{L,\chi}\leq\pi_{L_q,\alpha|_{L_q} }}}\pi_{L,\chi}\; ,
    \end{align}
    defines a noncontextual property for $\ConSS$. Equivalently,
    \begin{align}\label{eq: q-frame function}
        f_{q,\alpha}(\pi_{L,\chi})
        &=\begin{cases}
            1&\mathrm{if}\ \chi|_{L_q}=\alpha|_{L_q}\\
            0&\mathrm{if}\ \chi|_{L_q}\neq\alpha|_{L_q}
        \end{cases}\; ,
    \end{align}
    defines a Boolean-valued frame function $f:\StabSS\ra\{0,1\}$.
\end{lemma}

\begin{proof}
    Since (the contexts of) $\ConPP$ already contain all elements $\pi_{U,\chi}$ (see Eq.~(\ref{eq: symplectic projectors})), a noncontextual property on $\ConPP$ extends to one in $\ConSS$ if and only if every maximal context $C\in\cC_\mathrm{max}(\cP^\mathrm{symp}_n)$ contains exactly two elements $\pi_{L,\chi},\pi_{L',\chi'}$ (corresponding to abstract stabiliser projectors $\pi_{L,\chi}$ whose sum then defines the desired element).
    
    To see that this is the case, let $q:V\ra\zz_2$ be a quadratic refinement of Witt index $n-1$ and $\alpha:V\ra\zz_2$ a linear functional. For every Lagrangian subspace $L\in\Lag(V)$, $q$ defines a subspace $L_q<L$ of dimension $n-1$ by $L_q:=\ker(q|_L)$. A noncontextual property on $\ConPP$ thus defines a frame function $f_{q,\alpha}:\StabSS\ra\{0,1\}$ (over $\ConSS$), given by
    \begin{align*}
        f_{q,\alpha}(\pi_{L,\chi})
        =\begin{cases}
            1&\mathrm{if}\ \chi|_{L_q}=\alpha|_{L_q}\\
            0&\mathrm{otherwise}
        \end{cases}\; .
    \end{align*}
    Clearly, $f_{q,\alpha}$ selects two elements in every Lagrangian subspace, hence, in every maximal context of $\ConPP$. To see that this pertains in all stabiliser contexts, let
    \begin{align}\label{eq: q-level operators}
        t_{q,\alpha}
        :=\frac{1}{2^{n-1}}\sum_{\substack{v\in V\\q(v)=0}}(-1)^{\alpha(v)}e_v\; ,
    \end{align}
    and note that with Eq.~(\ref{eq: symplectic projectors}), we have
    \begin{align}\label{eq: weight-2 frame function as operator}
        \langle t_{q,\alpha},\pi_{L,\eta}\rangle
        =\frac{1}{2^{2n-1}}\sum_{\substack{v\in V\\q(v)=0}}
        \sum_{w\in L}(-1)^{\alpha(v)+\eta(w)}\langle e_v,e_w\rangle
        =\frac{1}{2^{n-1}}\sum_{w\in L_q}(-1)^{\alpha(w)+\eta(w)}
        =\begin{cases}
            1&\mathrm{if}\ \alpha|_{L_q}=\eta|_{L_q} \\
            0&\mathrm{otherwise}
        \end{cases}\; ,
    \end{align}
    where we used that $\{e_v\}_{v\in V}$ is an orthogonal basis of the abstract Bloch space $\R[V]$ in the second, and character orthogonality in the third step. Now, let $C\in\cC_\mathrm{max}(\cP^\mathrm{symp}_n)$ be any maximal stabiliser context with atoms $\{\pi_1,\cdots,\pi_{2^n}\}$. With Eq.~(\ref{eq: weight-2 frame function as operator}), any potential contextual dependence (on the choice of stabiliser context in the assignment of values of $f_{q,\alpha}$ to elements in $\pV$) disappears, and
    \begin{equation}\label{eq: weight 2}
        \sum_{i=1}^{2^n}f_{q,\alpha}(\pi_i)
        =\sum_{i=1}^{2^n}\langle t_{q,\alpha},\pi_i\rangle
        =\langle t_{q,\alpha},\sum_{i=1}^{2^n}\pi_i\rangle
        =\langle t_{q,\alpha},e_0\rangle
        =2\; ,
    \end{equation}
    since from $\langle\sum_{i=1}^{2^n}\pi_i,\sum_{i=1}^{2^n}\pi_i\rangle=2^n$, $\langle\sum_{i=1}^{2^n}\pi_i,e_0\rangle=2^n$ and $\langle e_0,e_0\rangle=2^n$, Cauchy-Schwarz implies that $\sum_{i=1}^{2^n}\pi_i=e_0$.

    For $n\geq3$, these exhaust the noncontextual properties of weight $2$ (see Lm.~\ref{lm: projective linear extension} below). For $n=2$, $\dim(L_q)=1$, hence, every map $s:q^{-1}(0)\to\{0,1\}$ with $s(0)=0$ restricts to a character on $L_q$, and the computation in Eq.~(\ref{eq: weight-2 frame function as operator}) then applies verbatim with $t_{q,\alpha}$ replaced by $t_{q,s}:=\frac{1}{2^{n-1}}\sum_{v\in q^{-1}(0)}(-1)^{s(v)}e_v$, yielding the frame functions $f_{q,s}$ of Prop.~\ref{prop: noncontextual MP-square}.
\end{proof}

The existence of generalised noncontextual properties may not seem surprising given that $\ConPP\cong(\Iso(V),\leq)$ admits valuations. However, the proofs of Lm.~\ref{lm: k=2} and Lm.~\ref{lm: noncontextual property in symplectic stabiliser poset} show that for the existence of a generalised noncontextual property it is not sufficient to only define a set of characters on contexts (here, isotropic subspaces), but one also needs a consistent coarse-graining of contexts. Quadratic refinements define such a coarse-graining into pairs of elements $\pi_{L,\chi}$ and $\pi_{L,\chi+q|_L}$. Yet, as we will show in the remainder of this section, valuations and the noncontextual properties in Lm.~\ref{lm: noncontextual property in symplectic stabiliser poset} are (up to complements) the only nontrivial such examples.

To prove this, we first need to show that the noncontextual properties in Lm.~\ref{lm: noncontextual property in symplectic stabiliser poset} are the only Boolean-valued frame functions of weight $2$. To this end, we will need two more concepts from projective geometry. The elliptic polar space $Q^-(2n-1,2)$ within $(V=\zz^{2n}_2,\omega)$ is the point-line geometry whose points are $\cQ^\circ=q^{-1}(0)\backslash\{0\}$ for $q$ a quadratic refinement of $\omega$ of Witt index $n-1$, and whose lines are the totally singular triples contained in $\cQ^\circ$. Moreover, we recall that a projective embedding of a point–line geometry $\Gamma$ is called universal if every other embedding of $\Gamma$ arises as a quotient of it; equivalently, every linear relation among the points of $\Gamma$ is generated by its line relations.

\begin{lemma}\label{lm: projective linear extension}
    Let $(V=\zz^{2n}_2,\omega)$ for $n\geq 3$ be a symplectic vector space. Let $q$ be a quadratic refinement of $\omega$ with Witt index $n-1$ and let $\cQ:=q^{-1}(0)$. Then every map $\widetilde{\alpha}:\cQ\ra\zz_2$ satisfying $\widetilde{\alpha}(v+w)=\widetilde{\alpha}(v)+\widetilde{\alpha}(w)$ (and $\widetilde{\alpha}(0)=0$) for every $v,w\in\cQ$ with $\omega(v,w)=0$ uniquely extends to a linear functional $\alpha:V\ra\zz_2$ with $\widetilde{\alpha}=\alpha|_\cQ$.
\end{lemma}

\begin{proof}
    The proof follows from universality of finite polar spaces (see Ref.~\cite{CardinaliGiuzziPasini2021}, see also Ref.~\cite{Tits1974}, chapter 8), in particular, of the point-line geometry $\Gamma=\cQ$ with points $\cQ^\circ:=q^{-1}(0)\backslash\{0\}$ and lines $\{v,w,v+w\}\in\cQ$. Explicitly, let $\zz^{\cQ^\circ}_2$ be the free vector space over $\cQ^\circ$ with basis elements $\{e_v\}_{v\in\cQ^\circ}$, and define
    \begin{align}\label{eq: universal embedding}
        E_\cQ
        :=\zz^{\cQ^\circ}_2/\langle e_v+e_w+e_{v+w}\mid v,w\in\cQ^\circ,v\neq w, \omega(v,w)=0\rangle\; ,
    \end{align}
    where the equivalence relation is with respect to all lines in $\cQ$. Now, define the canonical linear map $\Phi:E_\cQ\ra V$ by $\Phi([e_v])=v$, which is well-defined because every line relation is already an ambient linear relation in $V$ (since $v+w\in\cQ^\circ$ for all $v,w\in\cQ^\circ$ with $v\neq w$, $\omega(v,w)=0$). Surjectivity of $\Phi$ follows since $\cQ$ spans $V$, which is easily seen from the canonical form in Eq.~(\ref{eq: canonical form quadratic refinement}), and noting that $\cQ=q^{-1}(0)$ contains the symplectic basis elements $e_i,f_i$ for $i\geq 2$, as well as the elements $e_1+e_2+f_2$ and $f_1+e_2+f_2$ which combine to $e_1=(e_1+e_2+f_2)+e_2+f_2$ and $f_1=(f_1+e_2+f_2)+e_2+f_2$, hence, $\mathrm{span}_{\zz_2}(\cQ)=V$; injectivity follows from universality of the natural projective embedding of the elliptic polar space $Q^-(2n-1,2)$ (see Ref.~\cite{CardinaliGiuzziPasini2021}), whose point set corresponds with $\cQ^\circ$ and whose lines are the relations in Eq.~(\ref{eq: universal embedding}), by which every binary linear relation among its points is generated by its line relations.
    
    Finally, since by assumption $\widetilde\alpha:\cQ\ra\zz_2$ with $\widetilde\alpha(v+w)=\widetilde\alpha(v)+\widetilde\alpha(w)$, this map descends through the quotient, defining linear functionals $\overline{\alpha}:E_\cQ\ra\zz_2$, $\overline{\alpha}([e_v])=\widetilde{\alpha}(v)$ and, in turn, $\alpha:V\ra\zz_2$ given by $\alpha(v):=\overline{\alpha}\circ\Phi^{-1}(v)$.
\end{proof}

The proof of Lm.~\ref{lm: projective linear extension} does not apply to $n=2$, since in this case $\cQ^\circ$ is a $1$-ovoid, corresponding to a maximal set of mutually anti-commuting Pauli operators, which therefore does not contain any nontrivial relations.

\begin{lemma}\label{lm: k=2, symplectic}
    Let $f:\StabSS\ra\{0,1\}$ for $n\geq 3$ be a frame function of weight $2$. Then $f(\pi_{L,\chi})=\langle t_{q,\alpha},\pi_{L,\chi}\rangle$ for every Lagrangian subspace $L\in\Lag(V)$, character $\chi\in L^*$, and with $t_{q,\alpha}$ as in Eq.~(\ref{eq: q-level operators}).
\end{lemma}

\begin{proof}
    As with Lm.~\ref{lm: k=2}, the proof is split into two parts. The first part is entirely analogous to part 1 of Lm.~\ref{lm: k=2}, showing that any weight $k=2$ frame function corresponds to a quadratic refinement of Witt index $n-1$. However, unlike in the signed Pauli case, the set of abstract symplectic projectors $\pV$ does admit global assignments of characters.
    
    \textbf{Part 2.} Indeed, we show that the characters $\alpha_L$ glue to a linear functional $\alpha:V\ra\zz_2$. Since $X_L=\{\alpha_L,\alpha_L+q_L\}$, only the restriction of $\alpha_L$ to $\mathrm{ker}(q_L)$ is well-defined (but not the choice of $\alpha_L$ over $\alpha_L+q_L$). Since $q$ is a quadratic refinement with $i_W(q)=n-1$, we have $\dim(L_q)=n-1$, for $L_q:=\mathrm{ker}(q|_L)$. Now, define $\alpha_{L_q}:=\alpha_L|_{L_q}$, and note that
    \begin{align}
        \alpha_{L_q}|_{L_q\cap L'_q}
        =\alpha_{L'_q}|_{L_q\cap L'_q}\; .
    \end{align}
    In other words, $\{\alpha_{L_q}\}_{L\in\Lag(V)}$ defines a (noncontextual) compatible family of linear characters on the $0$-level set $\cQ:=q^{-1}(0)$. It follows from Lm.~\ref{lm: projective linear extension} that $\alpha_{L_q}=\alpha|_{L_q}$ for a unique linear functional $\alpha:V\ra\zz_2$, and thus
    \begin{align}
        X_L
        =\{\alpha|_L,\alpha|_L+q|_L\}\; ,
    \end{align}
    equivalently, $f=f_{q,\alpha}$ with $f_{q,\alpha}$ given by Eq.~(\ref{eq: q-frame function}).
\end{proof}

For $n=2$, Part 1 of Lm.~\ref{lm: k=2} still shows that every frame function of weight $2$ determines a quadratic refinement $q$ of Witt index $1$. In this case $q^{-1}(0)\backslash{0}$ is a $1$-ovoid, hence, carries no nontrivial compatibility relations; consequently the remaining character assignment is an arbitrary map $s:q^{-1}(0)\to\mathbb Z_2$ with $s(0)=0$. Thus, (after setting $\lambda_L=0$) the weight-$2$ symplectic frame functions are precisely the functions $f_{q,s}$ of Prop.~\ref{prop: noncontextual MP-square}. For $n\ge3$, Lm.~\ref{lm: k=2, symplectic} shows that these assignments are instead forced to arise from global linear functionals $\alpha:V\to\mathbb Z_2$.

Note that, up until this point, the distinction between the projective and signed Pauli group is entirely related with the group-theoretic extension problem
\begin{align}
    1 \ra \zz_4 \ra \cP_n \ra V \ra 1\; .
\end{align}
Indeed, the nontriviality of the cohomology class of the $2$-cocycle $\gamma_W:V\times V\ra\zz_4$ (and of $\beta_W$, for $n>1$) is the reason for (standard) contextuality of the Pauli group \cite{Raussendorf2019,RaussendorfEtAl2017,OkayTyhurstRaussendorf2018,OkaySheinbaum2019,RaussendorfEtAl2023}. By contrast, the nonexistence of noncontextual properties for higher-rank projectors in App.~\ref{app: classification} does not merely involve $\beta_W$, but is a consequence of the underlying symplectic structure of the Pauli group. In particular, note that Lm.~\ref{lm: character restriction}, Lm.~\ref{lm: Fourier coefficients} and Lm.~\ref{lm: reduced stab theory} all apply at the symplectic level; extending these arguments will rule out Boolean-valued frame functions that are not constant, weight-$1$ or of the form in Lm.~\ref{lm: k=2, symplectic}.\\

\textbf{Full classification.} We first consider frame functions of odd weight. We will again prove this by induction. The base case being $n=3$, for which the following lemmata assert that frame functions of weight $k=3,4,5$ do not exist.

\begin{lemma}\label{lm: n=3, k=3}
    There exists no frame function $f:\StabSS\ra\{0,1\}$ of weight $k=3,5$ for $n=3$.
\end{lemma}

\begin{proof}
    Let $k=3$, then the Fourier coefficients are either $a_v=\pm 1$ or $a_v=\pm 3$, by Lm.~\ref{lm: Fourier coefficients}. Now, define $O:=\{v\in V^\circ\mid |a_v|=3\}$. Using Parseval's identity, Eq.~(\ref{eq: Parseval identity}),
    \begin{align*}
        \sum_{v\in L^\circ}a^2_v
        =k(2^n-k)
        =15\; ,
    \end{align*}
    it follows that $|O\cap L^\circ|=1$ for every Lagrangian subspace $L\in\Lag(V)$, hence, $O$ is a $1$-ovoid. However, it is well-known that no $1$-ovoids exist for $n\geq 3$. The case $k=5$ follows by complementation, $\overline{f}=1-f$, from $k=3$.
\end{proof}

\begin{lemma}\label{lm: k=4, symplectic}
    There exists no frame function $f:\StabSS\ra\{0,1\}$ of weight $k=4$ for $n=3$.
\end{lemma}

\begin{proof}
    Let $f:\pV\ra\{0,1\}$ be the associated frame function of weight $k=4$. By Lm.~\ref{lm: Fourier coefficients}, we have $a_0=4$ and $a_v\in\{-4,-2,0,2,4\}$. Parseval's identity gives $\sum_{0\neq v\in L}a^2_v=k(2^n-k)=16$, which leaves two possibilities for the distribution of non-zero Fourier coefficients $\{a_v\}_{0\neq v\in L}$ for every $L\in\Lag(V)$: a single non-zero coefficient with $|a_v|=4$, or four non-zero coefficients with $|a_v|=2$. Let $A=\{0\neq v\in V\mid |a_v|=4\}$ and $B=\{0\neq v\in V\mid |a_v|=2\}$. Since there are $\prod^n_{i=1}(2^i+1)=135$ Lagrangian subspaces and every non-zero vector $v\in V$ lies in $\prod_{i=1}^{n-1}(2^i+1)=15$ of them \cite{AaronsonGottesman2004,Taylor1992},
    \begin{align}\label{eq: incidence count}
        15|A|+\frac{15}{4}|B|=135
        \quad\Longleftrightarrow\quad
        4|A|+|B|=36\; .
    \end{align}
    Clearly, for any two elements $a,a'\in A$, $\omega(a,a')=1$. These labels thus correspond to a set of mutually anti-commuting Pauli operators, whose cardinality is bounded by $|A|\leq2n+1=7$ (see e.g. Lm.~8 and Cor.~4 in Ref.~\cite{SarkarVanDenBerg2021}). Next, note that $B\subseteq\bigcap_{a\in A}\{b\in V^\circ\mid\omega(a,b)=1\}$ is an intersection of affine hyperplanes, and thus contains at most $2^{6-r}$ elements where $r=\mathrm{dim\ span}(A)$. Since the elements in $A$ are pairwise nonorthogonal, the restriction of $\omega$ to $\mathrm{span}(A)$ has Gram matrix with $0$'s on the diagonal and $1$'s in every off-diagonal entry. It follows that it has rank $|A|$ if $|A|$ is even and $|A|-1$ if $|A|$ is odd. Now assume that $|A|\geq 2$, then $r\geq 2$, hence, $B$ contains at most $16$ elements. Yet, this is incompatible with Eq.~(\ref{eq: incidence count}) thus ruling out $|A|=2,3,4$. Moreover, for $|A|\geq 5$, $r\geq 4$, hence, $B$ contains at most $4$ elements, yet from Eq.~(\ref{eq: incidence count}) $|B|\geq 8$ for $|A|\leq 7$. This leaves us with just two cases: $|A|\in\{0,1\}$.

    Let $\ta_v=\frac{1}{2}a_v\in\{\pm 1\}$ for every $v\in B$, and let $L\in\Lag(V)$ with $|B\cap L|=4$. Then $f_L(\chi)=\frac{1}{8}\sum_{v\in L}(-1)^{\chi(v)}a_v=\frac{1}{2}+\frac{1}{4}\sum_{v\in B\cap L}(-1)^{\chi(v)}\ta_v$ is Boolean and we have $\sum_{v\in B\cap L}(-1)^{\chi(v)}\ta_v=\pm 2$ for all $\chi\in L^*$. This condition is equivalent to
    \begin{align}\label{eq: bent relation}
        \prod_{v\in B\cap L}(-1)^{\chi(v)}\ta_v
        =(-1)^{\chi(\sum_{v\in B\cap L}v)}\prod_{v\in B\cap L}\ta_v
        =-1 \quad\quad\forall\chi\in L^*\; ,
    \end{align}
    and thus both $\sum_{v\in B\cap L}v=0$ and $\prod_{v\in B\cap L}\ta_v=-1$ (equivalently, $\sum_{v\in B\cap L}r_v=1$ for $\ta_v=(-1)^{r_v}$). We first extract structural information from the first condition, and then derive a contradiction from the second.

    If $|A|=0$, then $|B\cap L|=4$ for every Lagrangian $L\in\Lag(V)$, and $|B|=36$ by Eq.~(\ref{eq: incidence count}). Write $O:=V^\circ\backslash B$, so that $|O|=63-36=27$ and $|O\cap L|=3$ for every $L\in\Lag(V)$. We show that $O$ is an elliptic quadric.
    
    First, $O$ meets every Lagrangian subspace in a totally isotropic line. Indeed, the seven non-zero elements of $L$ sum to zero, hence so do the three elements of $O\cap L$, by the first condition in Eq.~(\ref{eq: bent relation}); yet, three distinct non-zero elements of $L$ summing to zero are necessarily of the form $\{v,w,v+w\}$, that is, they constitute a totally isotropic line.
    
    Next, define $q:V\to\zz_2$ by $q(0):=0$, $q|_O:=0$ and $q|_B:=1$, and let $v\neq w\in V^\circ$ with $\omega(v,w)=0$. Then $\{v,w,v+w\}$ is a totally isotropic line and thus contained in some $L\in\Lag(V)$. Since two distinct lines of the projective plane $L$ meet in exactly one point, $|\{v,w,v+w\}\cap(O\cap L)|\in\{1,3\}$, hence, $|\{v,w,v+w\}\cap B|\in\{0,2\}$ is even and therefore
    \begin{align*}
        q(v)+q(w)+q(v+w)
        =0=\omega(v,w)\; .
    \end{align*}
    In other words, $q$ is additive on commuting pairs, and it follows by Lm.~\ref{lm: compatible characters} (equivalently, by part 1 of the proof of Lm.~\ref{lm: k=2}) that $q=\alpha+cq_W$ for a linear functional $\alpha:V\to\zz_2$ and some $c\in\zz_2$. If $c=0$, then $q$ is linear and $|O|=|q^{-1}(0)|-1\in\{31,63\}$. If $c=1$, then $q$ is a quadratic refinement of $\omega$ and $|O|=|q^{-1}(0)|-1=2^{2n-1}\pm 2^{n-1}-1\in\{35,27\}$, according to whether $q$ has Witt index $i_W(q)=n$ or $i_W(q)=n-1$. Since $|O|=27$, this implies $c=1$ and $i_W(q)=n-1$. Consequently, $O=q^{-1}(0)\backslash\{0\}$ is an elliptic quadric $Q^-(5,2)$, and up to a symplectic transformation we may assume that $q$ is of the canonical form $q(v)=\sum_{i=1}^{3}a_ib_i+a_1+b_1$ in Eq.~(\ref{eq: canonical form quadratic refinement}).

    If $|A|=1$ with $a\in A$, we have $B=\{v\in V\mid\omega(v,a)=1\}$, hence, $B\cap L=\{v\in L\mid \omega(v,a)=1\}$ is an affine hyperplane in every Lagrangian subspace $L\in\Lag(V)$ with $a\notin L$, and thus $|B\cap L|=4$ (and $|B\cap L|=0$ otherwise).

    In both cases, $|A|=0$ and $|A|=1$, the first condition of Eq.~(\ref{eq: bent relation}) therefore holds, since there $B\cap L$ is an affine subspace and thus $\sum_{v\in B\cap L}v=0$. In order to arrive at a contradiction, we instead find a set of Lagrangian subspaces of odd cardinality such that every element of $B$ appears in an even number of them (we will highlight this below by colouring equal elements across the respective Lagrangians). In this case,
    \begin{align*}
        \sum_{i=1}^m\sum_{v\in B\cap L_i}r_v=0 \mod 2\; ,
    \end{align*}
    contradicting Eq.~(\ref{eq: bent relation}). To finish the proof, we provide explicit sets of such Lagrangians in the remaining two cases.\\
    
    For the elliptic quadratic $q(v)=\sum_{i=1}^3a_ib_i+a_1+b_1$ (cf. Eq.~(\ref{eq: canonical form quadratic refinement})) with $|A|=0$, let
    \begin{align*}
        L_1&=\langle e_1+f_1,\ e_3,\ f_2\rangle \\
        L_2&=\langle e_1+f_1,\ e_2+f_3,\ e_3+f_2+f_3\rangle \\
        L_3&=\langle e_1+f_1+f_2,\ e_2+f_1+f_2+f_3,\ e_3+f_2+f_3\rangle \\
        L_4&=\langle e_1+e_3+f_2+f_3,\ e_2+e_3+f_2,\ f_1+f_2+f_3\rangle \\
        L_5&=\langle e_1+e_3+f_3,\ e_2+e_3+f_2+f_3,\ f_1+f_2+f_3\rangle\; ,
    \end{align*}
    and note that these define Lagrangian subspaces, whose elements read
    \begin{align*}
        L_1&=\{0,e_1+f_1,e_3,e_1+e_3+f_1,f_2,e_1+f_1+f_2,e_3+f_2,e_1+e_3+f_1+f_2\} \\
        L_2&=\{0,e_1+f_1,e_2+f_3,e_1+e_2+f_1+f_3,e_3+f_2+f_3,e_1+e_3+f_1+f_2+f_3,e_2+e_3+f_2,e_1+e_2+e_3+f_1+f_2\} \\
        L_3&=\{0,e_1+f_1+f_2,e_2+f_1+f_2+f_3,e_1+e_2+f_3,e_3+f_2+f_3,e_1+e_3+f_1+f_3,e_2+e_3+f_1,e_1+e_2+e_3+f_2\} \\
        L_4&=\{0,e_1+e_3+f_2+f_3,e_2+e_3+f_2,e_1+e_2+f_3,f_1+f_2+f_3,e_1+e_3+f_1,e_2+e_3+f_1+f_3,e_1+e_2+f_1+f_2\} \\
        L_5&=\{0,e_1+e_3+f_3,e_2+e_3+f_2+f_3,e_1+e_2+f_2,f_1+f_2+f_3,e_1+e_3+f_1+f_2,e_2+e_3+f_1,e_1+e_2+f_1+f_3\}\; ,
    \end{align*}
    and whose intersection with $B$ is given by
    \begin{align*}
        B\cap L_1
        &=\{
        \textcolor{cH}{e_1+f_1},\
        \textcolor{cE}{e_1+e_3+f_1},\
        \textcolor{cJ}{e_1+f_1+f_2},\
        \textcolor{cB}{e_1+e_3+f_1+f_2}
        \}\\
        B\cap L_2
        &=\{
        \textcolor{cH}{e_1+f_1},\
        \textcolor{cC}{e_1+e_2+f_1+f_3},\
        \textcolor{cI}{e_3+f_2+f_3},\
        \textcolor{cF}{e_2+e_3+f_2}
        \}\\
        B\cap L_3
        &=\{
        \textcolor{cJ}{e_1+f_1+f_2},\
        \textcolor{cG}{e_1+e_2+f_3},\
        \textcolor{cI}{e_3+f_2+f_3},\
        \textcolor{cA}{e_2+e_3+f_1}
        \}\\
        B\cap L_4
        &=\{
        \textcolor{cF}{e_2+e_3+f_2},\
        \textcolor{cG}{e_1+e_2+f_3},\
        \textcolor{cD}{f_1+f_2+f_3},\
        \textcolor{cE}{e_1+e_3+f_1}
        \}\\
        B\cap L_5
        &=\{
        \textcolor{cA}{e_2+e_3+f_1},\
        \textcolor{cB}{e_1+e_3+f_1+f_2},\
        \textcolor{cC}{e_1+e_2+f_1+f_3},\
        \textcolor{cD}{f_1+f_2+f_3}
        \}\; .
    \end{align*}
    
    Finally, for the case $|A|=1$, where without loss of generality $e_1\in A$, consider the following five Lagrangians
    \begin{align*}
        L_1&=\langle e_1+e_2+f_2+f_3,\ e_3+f_2+f_3,\ f_1+f_2\rangle\\
        L_2&=\langle e_1+f_2+f_3,\ e_2+f_1+f_2,\ e_3+f_1\rangle \\
        L_3&=\langle e_1+f_3,\ e_2+f_3,\ e_3+f_1+f_2\rangle \\
        L_4&=\langle e_2+f_2,\ e_3+f_3,\ f_1\rangle \\
        L_5&=\langle e_3,\ f_1,\ f_2\rangle\; ,
    \end{align*}
    whose elements are given explicitly as follows
    \begin{align*}
        L_1&=\{0,e_1+e_2+f_2+f_3,e_3+f_2+f_3,e_1+e_2+e_3,f_1+f_2,e_1+e_2+f_1+f_3,e_3+f_1+f_3,e_1+e_2+e_3+f_1+f_2\} \\
        L_2&=\{0,e_1+f_2+f_3,e_2+f_1+f_2,e_1+e_2+f_1+f_3,e_3+f_1,e_1+e_3+f_1+f_2+f_3,e_2+e_3+f_2,e_1+e_2+e_3+f_3\} \\
        L_3&=\{0,e_1+f_3,e_2+f_3,e_1+e_2,e_3+f_1+f_2,e_1+e_3+f_1+f_2+f_3,e_2+e_3+f_1+f_2+f_3,e_1+e_2+e_3+f_1+f_2\} \\
        L_4&=\{0,e_2+f_2,e_3+f_3,e_2+e_3+f_2+f_3,f_1,e_2+f_1+f_2,e_3+f_1+f_3,e_2+e_3+f_1+f_2+f_3\} \\
        L_5&=\{0,e_3,f_1,e_3+f_1,f_2,e_3+f_2,f_1+f_2,e_3+f_1+f_2\}\; ,
    \end{align*}
    and whose intersection with $B$ is given by
    \begin{align*}
        B\cap L_1
        &=\{
        \textcolor{cA}{f_1+f_2},\
        \textcolor{cC}{e_1+e_2+f_1+f_3},\
        \textcolor{cD}{e_3+f_1+f_3},\
        \textcolor{cB}{e_1+e_2+e_3+f_1+f_2}
        \} \\
        B\cap L_2
        &=\{
        \textcolor{cJ}{e_2+f_1+f_2},\
        \textcolor{cC}{e_1+e_2+f_1+f_3},\
        \textcolor{cF}{e_3+f_1},\
        \textcolor{cH}{e_1+e_3+f_1+f_2+f_3}
        \} \\
        B\cap L_3
        &=\{
        \textcolor{cG}{e_3+f_1+f_2},\
        \textcolor{cH}{e_1+e_3+f_1+f_2+f_3},\
        \textcolor{cI}{e_2+e_3+f_1+f_2+f_3},\
        \textcolor{cB}{e_1+e_2+e_3+f_1+f_2}
        \} \\
        B\cap L_4
        &=\{
        \textcolor{cE}{f_1},\
        \textcolor{cJ}{e_2+f_1+f_2},\
        \textcolor{cD}{e_3+f_1+f_3},\
        \textcolor{cI}{e_2+e_3+f_1+f_2+f_3}
        \} \\
        B\cap L_5
        &=\{
        \textcolor{cE}{f_1},\
        \textcolor{cF}{e_3+f_1},\
        \textcolor{cA}{f_1+f_2},\
        \textcolor{cG}{e_3+f_1+f_2}
        \}\; .
    \end{align*}
    In all cases the elements in these Lagrangian subspaces that are also elements of $B$ appear exactly twice (as highlighted by the colouring of these elements in the respective cases). Consequently,
    \begin{align*}
        \sum_{i=1}^5\sum_{v\in B\cap L_i}r_v=0 \mod 2\; ,
    \end{align*}
    contradicting Eq.~(\ref{eq: bent relation}). We thus conclude that for $n=3$ no Boolean-valued frame function of weight $k=4$ exists.
\end{proof}

Next, we apply an induction argument to lift the results for $n=3$ (in Lm.~\ref{lm: k=2, symplectic}, Lm.~\ref{lm: n=3, k=3} and Lm.~\ref{lm: k=4, symplectic}) to rule out frame functions not corresponding to linear functionals or quadratic refinements as in Lm.~\ref{lm: k=2, symplectic}.

\begin{lemma}\label{lm: weights for symplectic Boolean-valued frame functions}
    Let $f:\StabSS\ra\{0,1\}$ be a frame function for $n\geq 3$. Then $k\in\{0,1,2,2^n-2,2^n-1,2^n\}$.
\end{lemma}

\begin{proof}
    The assertion holds for $n=3$ by Lm.~\ref{lm: k=2, symplectic}, Lm.~\ref{lm: n=3, k=3} and Lm.~\ref{lm: k=4, symplectic}. Now, let $f:\StabSS\ra\{0,1\}$ be frame function for $n\geq 4$. Let $0\neq v\in V$ and consider the eigenspaces corresponding to the vectors $\pi^\pm_v:=\frac{1}{2}(e_0\pm e_v)$, which act as projectors onto elements $\pi_{L,\chi}\in\StabSS$ with $\langle e_v,\pi_{L,\chi}\rangle=\pm 1$, equivalently, onto all elements $\pi_{L,\chi}$ with $v\in L$ and $(-1)^{\chi(v)}=\pm 1$. Noting that Lm.~\ref{lm: reduced stab theory} applies verbatim at the symplectic level (with $\StabS$ replaced by $\StabSS$ and $\Pi^\pm_v\StabS\cong\cS^\mathrm{stab}_{n-1}$ by $\pi^\pm_v\StabSS\cong\cS^\mathrm{symp}_{n-1}$), $f$ restricts to Boolean-valued frame functions on $\pi^\pm_v\StabSS$ of weights $k^\pm_v=f(\pi^\pm_v)$. Since $k=k^+_v+k^-_v$, the inductive hypothesis implies that we need to consider weights $k\in\{0,1,2,3,4,2^{n-1}-2,2^{n-1}-1,2^{n-1}\}$ (since, up to complementation, it is sufficient to consider weights $k\leq 2^{n-1}$). We rule out the cases $k\in\{3,4,2^{n-1}-2,2^{n-1}-1,2^{n-1}\}$ one by one, using again the relations $a_v=k^+_v-k^-_v$ and $\sum_{0\neq v\in L}a^2_v=k(2^n-k)$ for every Lagrangian subspace $L\in\Lag(V)$.

    $\mathbf{k=3.}$ In this case, the only possibilities are $(k^+_v,k^-_v)=(1,2)$, or $(k^+_v,k^-_v)=(2,1)$ and thus $a_v=\pm 1$ for all $0\neq v\in V$, which yields a contradiction with Parseval's identity $\sum_{0\neq v\in L}a^2_v=2^n-1<3(2^n-3)$ for $n\geq 3$.

    $\mathbf{k=4.}$ In this case, the only possibility is $(k_v^+,k_v^-)=(2,2)$, and thus $a_v=k_v^+-k_v^-=0$ for every $0\neq v\in V$. However, in this case for every Lagrangian $L\in\Lag(V)$ and every character $\chi\in L^*$ we obtain
    \begin{align*}
        f(\pi_{L,\chi})
        =\frac{1}{2^n}\sum_{v\in L}(-1)^{\chi(v)}a_v
        =\frac{4}{2^n}\; .
    \end{align*}
    which for $n\geq 4$ is clearly not Boolean-valued. Consequently, no weight-$4$ frame function exists for $n\geq 4$.

    $\mathbf{k=2^{n-1}-2.}$ In this case, Parseval's identity, Eq.~(\ref{eq: Parseval identity}), yields
    \begin{align*}
        \sum_{0\neq v\in L}a^2_v
        =k(2^n-k)
        =(2^{n-1}-2)(2^n-(2^{n-1}-2))
        =4^{n-1}-4\; .
    \end{align*}
    On the other hand, using the inductive hypothesis and $k^+_v+k^-_v=k$, the only possibilities are $(k^+_v,k^-_v)=(0,2^{n-1}-2)$ or $(k^+_v,k^-_v)=(2^{n-1}-2,0)$ and thus $a_v=k^+_v-k^-_v=\pm(2^{n-1}-2)$. But since every Lagrangian subspace has $2^n-1$ non-zero elements, and each contributes $a^2_v=(2^{n-1}-2)^2$, one computes
    \begin{align*}
        \sum_{0\neq v\in L}a^2_v
        =(2^n-1)(2^{n-1}-2)^2\; .
    \end{align*}
    Clearly, the two expressions differ for $n\geq 4$, hence, this case is impossible.

    $\mathbf{k=2^{n-1}-1.}$ In this case, the possible pairs are $(k^+_v,k^-_v)\in\{(0,2^{n-1}-1),(1,2^{n-1}-2),(2^{n-1}-2,1),(2^{n-1}-1,0)\}$, hence, $a_v\in\{\pm(2^{n-1}-1),\pm(2^{n-1}-3)\}$ and thus $a^2_v\geq (2^{n-1}-3)^2$. Yet, this contradicts Parseval's identity for $n\geq 4$,
    \begin{align*}
        (2^n-1)(2^{n-1}-3)^2
        \leq\sum_{0\neq v\in L}a^2_v
        =(2^{n-1}-1)(2^{n-1}+1)
        =4^{n-1}-1\; .
    \end{align*}

    $\mathbf{k=2^{n-1}.}$ In this case, Parseval's identity, Eq.~(\ref{eq: Parseval identity}), yields
    \begin{align*}
        \sum_{0\neq v\in L}a^2_v
        =k(2^n-k)
        =4^{n-1}\; .
    \end{align*}
    Yet, the only possibilities are $(k^+_v,k^-_v)\in\{(0,2^{n-1}),(1,2^{n-1}-1),(2,2^{n-1}-2),(2^{n-1}-2,2),(2^{n-1}-1,1),(2^{n-1},0)\}$, and thus $a_v=k^+_v-k^-_v\in\{\pm 2^{n-1},\pm (2^{n-1}-2),\pm(2^{n-1}-4)\}$, hence, $a^2_v\geq(2^{n-1}-4)^2$. Consequently, summing non-zero elements in a Lagrangian subspace, one arrives at
    \begin{align*}
        \sum_{0\neq v\in L}a^2_v
        \geq(2^n-1)(2^{n-1}-4)^2\; .
    \end{align*}
    This inequality is strict for $n\geq 4$, hence, this case is also excluded.
\end{proof}

\begin{proof}[Proof of Thm.~\ref{thm: symplectic Boolean frame functions}]
    By Lm.~\ref{lm: weights for symplectic Boolean-valued frame functions}, the only weights for Boolean-valued frame functions are $0,1,2,2^n-2,2^n-1,2^n$. The edge cases correspond with constant frame functions, $k=1$ (and by complementation $k=2^n-1$) with families of compatible characters $\{\chi_L\}_{L\in\Lag(V)}$ corresponding either to a linear functional $\alpha:V\ra\zz_2$ or a quadratic refinement $q:V\ra\zz_2$ by Lm.~\ref{lm: compatible characters}, and $k=2$ (and by complementation $k=2^n-2$) with linear characters supported on the level sets of quadratic refinements of Witt index $n-1$ by Lm.~\ref{lm: noncontextual property in symplectic stabiliser poset} and Lm.~\ref{lm: k=2, symplectic}.

    Finally, note that the classification above uses finite additivity only with respect to the contexts $c_U$ with $U\in\Iso(V)$, that is, with respect to $\ConPP\subset\ConSS$: this is immediate for Parseval's identity in Eq.~(\ref{eq: Parseval identity}), Lm.~\ref{lm: Fourier coefficients} and Lm.~\ref{lm: reduced stab theory} (see remark following its proof), as well as for Lm.~\ref{lm: k=2, symplectic}--Lm.~\ref{lm: weights for symplectic Boolean-valued frame functions}. Conversely, the frame functions $f_{\alpha,c}$ and $f_{q,\alpha}$ are additive on \emph{every} maximal context of $\ConSS$, by the computations following Lm.~\ref{lm: compatible characters} and by Lm.~\ref{lm: noncontextual property in symplectic stabiliser poset}, respectively. Consequently, (nonconstant) Boolean-valued frame functions over $\ConPP$ and over $\ConSS$ coincide.
\end{proof}

In summary, the context posets $\ConP$ and $\ConPP$, which as posets are both isomorphic (to $\Iso(V)$), are distinguished by the concrete identifications of their Boolean algebras along inclusions (that is, as diagrams). In general, these correspond with different choices of families $\{\lambda_U\}_{U\in\Iso(V)}$ parametrising a given $2$-cocycle $\beta_W$; for the cases $\ConP$ and $\ConPP$, their respective cohomology classes are distinguished by $[\beta_W]\neq 0$ and $[\beta_0]=0$, respectively. This distinction has been identified as criterion for contextuality in Ref.~\cite{Raussendorf2019,RaussendorfEtAl2017,OkayTyhurstRaussendorf2018,OkaySheinbaum2019,RaussendorfEtAl2023} - in the sense of the existence of valuations (see Sec.~\ref{sec: NC properties} and App.~\ref{app: frame functions}). Yet, Thm.~\ref{thm: symplectic Boolean frame functions} shows that this is not the only source of contextual behaviour of the $n$-qubit Pauli group, in the sense of Kochen and Specker \cite{Specker1960,KochenSpecker1967}. In fact, its reformulation in terms of context connections in Ref.~\cite{Frembs2024,Frembs2025} highlights that the nonexistence of valuations is a test of contextuality, but its existence not one of noncontextuality. Instead, noncontextuality also implies the existence of other noncontextual properties as defined in Def.~\ref{def: noncontextual property}. By analysing those, Thm.~\ref{thm: symplectic Boolean frame functions} proves that the cohomology-reduced symplectic theory, too, is contextual. As a consequence, the cohomological characterisation in Ref.~\cite{Raussendorf2019,RaussendorfEtAl2017,OkayTyhurstRaussendorf2018,OkaySheinbaum2019,RaussendorfEtAl2023}, too, is generally a test of contextuality only, but not of noncontextuality.

\section{On Gleason's theorem for stabiliser theory}\label{app: Gleason}

We relate our characterisation of stabiliser frame functions to Gleason's theorem \cite{Gleason1957} for the restricted context poset $\ConP,\ConS\subset\CH$ and the abstract poset $\ConSS$. Recall that Gleason's theorem asserts that every frame function $f:\PH\ra\R_+$ for $\dim(\cH)\geq 3$ is represented by a positive Hermitian operator $a\in\LHsa$, via the Born rule
\begin{align*}
    f(p)=\tr[pa]\; .
\end{align*}
For the restricted stabiliser poset, $\ConS$, one instead has the following result.

\begin{lemma}\label{lm: stabiliser Gleason}
    Let $f:\StabS\ra\R_+$ be a frame function. Then there exists a Hermitian operator $T\in\cL_\mathrm{sa}(\C^{2^n})$ such that
    \begin{align*}
        f(\psi)
        =\tr[T\Pi_\psi]
        \quad\quad\forall\psi\in\StabS\; .
    \end{align*}
\end{lemma}

\begin{proof}
    Let $V=\zz^{2n}_2$, and for every $L\in\Lag(V)$, set $f_L:=f|_L$, $b_{v,L}:=\hf_L(v)$ and $a_v:=(-1)^{\lambda_L(v)}b_{v,L}$ as in Eq.~(\ref{eq: L-restricted Fourier coefficients}). Now, define $T=\frac{1}{2^n}\sum_{v\in V}a_vW_v$, and note that Eq.~(\ref{eq: symplectic codespace projectors}), character orthogonality and Fourier inversion imply
    \begin{equation*}
        \tr[T\Pi_{L,\chi}]
        =\frac{1}{2^{2n}}\sum_{v\in V}\sum_{w\in L}(-1)^{\chi(w)+\lambda_L(w)}a_v\tr[W_vW_w]
        =\frac{1}{2^n}\sum_{v\in L}(-1)^{\chi(v)+\lambda_L(v)}a_v 
        =\frac{1}{2^n}\sum_{v\in L}(-1)^{\chi(v)}b_{v,L} 
        =f(\Pi_{L,\chi})\; .
    \end{equation*}
    Moreover, $T$ is unique since $\aP_n$ is a basis of $\cL_\mathrm{sa}(\C^{2^n})$.
\end{proof}

Note that unlike Gleason's theorem \cite{Gleason1957}, the operator $T$ in Lm.~\ref{lm: stabiliser Gleason} is generally not a density matrix, in particular, it is generally not positive. For instance, the operator $T=\frac{1}{2}(1+\sum_{0\neq v\in q^{-1}(0)}(-1)^{s(v)}W_v)$ corresponding to the frame function of Prop.~\ref{prop: noncontextual MP-square} is not positive, but has eigenvalues $\frac{1}{2}(1\pm\sqrt{5})$. However, it is positive on stabiliser states, hence, an operator $T$, which corresponds to a frame function of weight $k\neq 0$, defines an element $\frac{T}{k}$ of the $\Lambda$-polytope,
\begin{align*}
    \Lambda_n
    :=\{X\in\cL_\mathrm{sa}(\C^{2^n})\mid\tr[X]=1,\ \tr[X\Pi_\psi]\geq 0\ \forall\psi\in\StabS\}\; ,
\end{align*}
defined in Ref.~\cite{ZurelOkayRaussendorf2020} as the state space of a nonnegative hidden-variable representation which gives rise to a classical simulation algorithm for $n$-qubit quantum systems. In particular, the above operator $T$ (corresponding to the Boolean-valued frame function of Prop.~\ref{prop: noncontextual MP-square}) is Hermitian and has $\tr[T]=2$, hence, $\frac{T}{2}\in\Lambda_2$.

From the perspective of the $\Lambda$-polytope, Thm.~\ref{thm: stabiliser frame functions} thus rules out a certain class of operators.

\begin{corollary}\label{cor: Boolean stabiliser Gleason}
    Let $T\in\cL_\mathrm{sa}(\C^{2^n})$ for $n\geq 3$ be such that $\tr[T\Pi_\psi]\in\{0,1\}$ for all $\psi\in\StabS$, then $T\in\{0,\one\}$.
\end{corollary}

\begin{proof}
    Let $f_T(\psi):=\tr[T\Pi_\psi]$. For any stabiliser basis $B$ we have $\sum_{\psi\in B}f_T(\psi)=\sum_{\psi\in B}\tr[T\Pi_\psi]=\tr[T]$, hence, $f_T$ is a frame function. If $f_T(\psi)\in\{0,1\}$ for all $\psi\in\StabS$, then $f_T=c$ for $c\in\{0,1\}$ by Thm.~\ref{thm: stabiliser frame functions}. Hence, $\tr[(T-c\one)\Pi_\psi]=0$ for all $\psi\in\StabS$ and thus $T=c\one$ by uniqueness in Lm.~\ref{lm: stabiliser Gleason}.
\end{proof}

Cor.~\ref{cor: Boolean stabiliser Gleason} generalises the fact that $\Lambda_n$ contains no globally outcome-deterministic ontic states $X\in\cL_1(\C^{2^n})$ for all stabiliser measurements and $n\geq 2$ - which holds as a consequence of standard contextuality arguments - to all operators $X\in\cL_\mathrm{sa}(\C^{2^n})$ without the restriction to unit trace whenever $n\geq 3$.\\

Finally, we note that a similar Gleason-type representation exists also in the symplectic case.

\begin{lemma}\label{lm: symplectic Gleason}
    Let $f:\StabSS\ra\R_+$ be a frame function. Then there exists a unique abstract operator $t\in\R[V]$ such that
    \begin{align}\label{eq: symplectic Gleason}
        f(\pi_{L,\chi})
        =\langle t,\pi_{L,\chi}\rangle\; .
    \end{align}
\end{lemma}

\begin{proof}
    The proof is analogous to that of Lm.~\ref{lm: stabiliser Gleason}, with $t=\frac{1}{2^n}\sum_{v\in V}a_ve_v$, the Hilbert-Schmidt inner product replaced by $\langle\cdot,\cdot\rangle$, and with $\lambda_L=0$ for all $L\in \Lag(V)$ and thus $a_v=b_{v,L}$. Uniqueness follows since the vectors $\pi_{L,\chi}$ span $\R[V]$, for $e_v=\pi^+_v-\pi^-_v$ and each $\pi^\pm_v$ is a sum of Lagrangian atoms by Lm.~\ref{lm: character restriction}.
\end{proof}

Indeed, the vectors $t_g$ and $t_{q,\alpha}$ in Eq.~(\ref{eq: weight-1 frame function operator}) and Eq.~(\ref{eq: q-level operators}), respectively, are of the form in Eq.~(\ref{eq: symplectic Gleason}).

\section{On the relation with Cameron-Liebler sets and qudit stabiliser theory}\label{app: CL correspondence}

The abstract symplectic stabiliser theory considered above admits a natural interpretation in finite affine geometry. We briefly make this connection explicit and relate Thm.~\ref{thm: symplectic Boolean frame functions} to the Cameron--Liebler problem studied by Guo and Wan in Ref.~\cite{GuoWan2023}, and discuss the generalisation to the qudit case.\\

\textbf{Classification of binary affine-symplectic Cameron--Liebler sets.} Let $(V=\mathbb F_2^{2n},\omega)$ and denote by
$\cO_n$ the set of maximal totally isotropic flats, that is, the set of affine Lagrangians
\begin{align*}
   \cO_n
   =\{a+L\mid L\in\Lag(V),\ a\in V\}\; ,
\end{align*}
where $a+L$ denotes a coset of a Lagrangian subspace $L$. For every $L\in\Lag(V)$, nondegeneracy of $\omega$ induces an
isomorphism $V/L \stackrel{\sim}{\longrightarrow} L^*$ given by $a+L\longmapsto \chi_a^L$ with $\chi_a^L(v):=\omega(a,v)$. Consequently, there is a bijection
\begin{equation}\label{eq: affine bijection}
   \Phi:S_n^{\rm symp}\stackrel{\sim}{\longrightarrow}\cO_n\; ,\qquad
   \Phi(\pi_{L,\chi_a^L})=a+L\; \; ,
\end{equation}
and this bijection preserves the relevant incompatibility relation. Indeed, for $L,M\in\Lag(V)$,
\begin{align*}
   (a+L)\cap(b+M)\neq\emptyset
   \Longleftrightarrow
   a-b\in L+M
   \Longleftrightarrow
   \omega(a-b,v)=0\quad\forall\,v\in L\cap M\; ,
\end{align*}
where we used $(L+M)^\perp=L^\perp\cap M^\perp=L\cap M$. Together with Eq.~\eqref{eq: stabiliser orthogonality}, this yields
\begin{equation}\label{eq: orthogonality vs disjointness}
   \pi_{L,\chi_a^L}\perp\pi_{M,\chi_b^M}
   \iff
   (a+L)\cap(b+M)=\emptyset\; .
\end{equation}
A set of $2^n$ pairwise orthogonal elements of $S_n^{\rm symp}$ therefore corresponds to $2^n$ pairwise disjoint affine Lagrangians. Since each contains $2^n$ points and $|V|=2^{2n}$, such a set partitions $V$. Thus maximal contexts in $C(P_n^{\rm symp})$ correspond precisely to maximal totally isotropic spreads in the terminology of Ref.~\cite{GuoWan2023}.

Let $M$ denote the point--flat incidence matrix, defined for $x\in V$ and $F\in\cO_n$ by
\begin{align*}
   M_{x,F}=
   \begin{cases}
      1&x\in F\\
      0&x\notin F
   \end{cases}\; .
\end{align*}
Ref.~\cite{GuoWan2023} calls a subset $\cL\subseteq\cO_n$ a \emph{Cameron--Liebler set with parameter $x$} if its characteristic vector satisfies
\begin{align*}
   \chi_\cL\in\mathrm{Im}(M^\mathsf T)\; ,\qquad
   |\cL|
   =x\prod_{i=1}^n(2^i+1)\; .
\end{align*}
By Ref.~\cite[Thm.~4.5]{GuoWan2023}, in the symplectic case this is equivalent to
\begin{equation}\label{eq: x}
   |\cL\cap\cS|=x\qquad
   \text{for every maximal totally isotropic spread }\cS\; .
\end{equation}
It follows immediately from Eqs.~(\ref{eq: affine bijection})--(\ref{eq: x}) that Boolean frame functions are precisely Cameron--Liebler sets.

\begin{proposition}\label{prop: frame function - Cameron-Liebler correspondence}
    The correspondence
    \begin{align*}
       f\longmapsto
       \cL_f:=
       \{\Phi(\pi)\mid \pi\in S_n^{\rm symp},\,f(\pi)=1\}
    \end{align*}
    defines a bijection between Boolean-valued frame functions $f:S_n^{\rm symp}\to\{0,1\}$ of weight $k$ and Cameron--Liebler sets of maximal totally isotropic flats with parameter $x=k$.
\end{proposition}

\begin{proof}
    By Eq.~\eqref{eq: orthogonality vs disjointness}, maximal stabiliser contexts correspond exactly to maximal totally isotropic spreads, and, using Eq.~(\ref{eq: x}) (which holds by Ref.~\cite[Thm.~4.5]{GuoWan2023}), the frame-function condition $\sum_{\pi\in C}f(\pi)=k$ for every maximal context $C\in\ConSS$ is equivalent to $|\cL_f\cap\cS|=k$ for every spread $\cS$.
\end{proof}

In other words, for every $L\in\Lag(V)$ its $2^n$ cosets form a spread; hence, a weight-$k$ frame function selects $k$ cosets for each of the $|\Lag(V)|=\prod_{i=1}^n(2^i+1)$ Lagrangian subspaces, and consequently $|\cL_f|=k\prod_{i=1}^n(2^i+1)$. We may therefore restate Thm.~\ref{thm: symplectic Boolean frame functions} as a complete classification of binary affine-symplectic Cameron--Liebler sets.

\begin{corollary}\label{cor: binary-CL}
    Let $n\ge3$ and let $\cL$ be a Cameron-Liebler set of maximal totally isotropic flats with parameter $x$. Then $x\in\{0,1,2,2^n-2,2^n-1,2^n\}$. Up to complementation, and under the correspondence in Eq.~(\ref{eq: affine bijection}), every nontrivial such set is the support of one of the frame functions $f_{\alpha,c}$ or $f_{q,\alpha}$ of Thm.~\ref{thm: symplectic Boolean frame functions}.
    
    In particular, there are no other Cameron--Liebler sets of maximal totally isotropic flats over $\mathbb F_2$ for $n\ge3$.
\end{corollary}

Together with the $n=2$ classification above and the elementary $n=1$ case, this completes the binary classification in all ranks. For example, write $\alpha(v)=\omega(a_\alpha,v)$. The families corresponding to $f_{\alpha,0}$ are precisely the point-pencils
\begin{align*}
   \mathcal P_{a_\alpha}
   =\{a_\alpha+L\mid L\in\Lag(V)\}\; .
\end{align*}
If $q$ is a quadratic refinement of $\omega$, the family corresponding to the valuation $f_q$ may equivalently be written
\begin{align*}
   \cQ_q
   :=
   \{a+L\in\cO_n
      \mid q\text{ is constant on }a+L\},
\end{align*}
since $q(a+v)=q(a)+q(v)+\omega(a,v)$ is independent of $v\in L$ precisely when $\omega(a,\cdot)|_L=q|_L$. Thus parameter-one Cameron--Liebler sets consist of point-pencils and the quadratic-refinement families $\cQ_q$. Similarly, for a quadratic refinement $q$ of Witt index $n-1$, the parameter-two family associated with $f_{q,\alpha}$ in Thm.~\ref{thm: symplectic Boolean frame functions} is
\begin{align*}
    \cL_{q,\alpha}
    =\{a+L\mid\omega(a,\cdot)|_{L_q}
    =\alpha|_{L_q}\}\; ,\qquad
    L_q=L\cap q^{-1}(0),
\end{align*}
and its complement gives the parameter $2^n-2$ case.\\

\textbf{Odd-prime qudits and the general Cameron--Liebler classification problem.} The preceding correspondence places the binary result studied in this work within its natural generalisation to qudit stabiliser theory. Let now $V=\mathbb F_p^{2n}$ for an odd prime $p$. In this case, one may choose the standard Weyl representation (cf. Eq.~(\ref{eq: Weyl operators})) such that the multiplication cocycle vanishes on isotropic subspaces: since $2$ is invertible in $\mathbb F_p$,
\begin{align*}
    W_vW_w=\zeta^{\frac12\omega(v,w)}W_{v+w}\; ,\qquad \zeta=e^{2\pi i/p}\; ,
\end{align*}
hence, $W_vW_w=W_{v+w}$ whenever $\omega(v,w)=0$. Consequently, pure $n$-qudit stabiliser states are naturally parametrised by affine Lagrangian subspaces $a+L\subset V$, with $L\in\Lag(V)$ \cite{Gross2006,deBeaudrap2013}. Equivalently, the character $\chi\in L^*$ labelling a joint stabiliser eigenspace is uniquely represented by a coset $a+L$ through $\chi(v)=\omega(a,v)$ for $v\in L$. Orthogonality of stabiliser states is then equivalent to disjointness of the corresponding affine Lagrangians. Thus, unlike for qubits, where the nontrivial Pauli cocycle distinguishes the signed stabiliser theory from the abstract symplectic theory, for odd $p$ the Cameron--Liebler description applies directly to the stabiliser states themselves.

In particular, Prop.~\ref{prop: frame function - Cameron-Liebler correspondence} extends verbatim with $\mathbb F_2$ replaced by $\mathbb F_p$: Boolean-valued stabiliser frame functions of weight $k$ are in bijection with Cameron--Liebler sets of maximal totally isotropic flats with parameter $x=k$. Hence, a complete characterisation of noncontextual properties for odd-prime stabiliser theory is equivalent to the corresponding Cameron--Liebler classification problem.  Guo and Wan formulate this problem over arbitrary finite fields and establish several equivalent characterisations and partial classification results \cite{GuoWan2023}; Thm.~\ref{thm: symplectic Boolean frame functions} solves the binary case (via Cor.~\ref{cor: binary-CL}), yet a general classification of affine symplectic Cameron--Liebler sets is not presently known.

We end by commenting on the special rigidity of the binary case. In this case, the Fourier coefficients associated with a Boolean frame function are integers subject to strong parity restrictions, and restriction to a Pauli eigenspace produces only two lower-dimensional frame functions. Moreover, characteristic two admits quadratic refinements $q(v+w)=q(v)+q(w)+\omega(v,w)$, whose two Witt types provide the decisive structure in the weight-one and weight-two classifications. None of these features persists in the same form for odd $p$: Fourier coefficients are sums of $p$-th roots of unity, restriction along a nonzero Weyl operator decomposes into $p$ eigenspaces, and a nonzero alternating symplectic form cannot arise as the polarisation of an ordinary quadratic form in odd characteristic. The resulting Boolean incidence problem therefore has substantially more local degrees of freedom. This is consistent with the broader theory of Boolean degree-one functions on finite classical association schemes, where polar-space classification problems are known to admit considerably richer behaviour and complete classifications are generally difficult \cite{FilmusIhringer2019}.

\end{document}